\documentclass[sigconf,screen]{acmart}

\graphicspath{{figures/}}

\usepackage{enumitem}
\usepackage[capitalize]{cleveref}
\crefname{proposition}{Proposition}{Propositions}
\crefname{theorem}{Theorem}{Theorems}
\crefname{equation}{Equation}{Equations}
\crefname{figure}{Figure}{Figures}
\crefname{table}{Table}{Tables}
\usepackage{algorithm}
\usepackage{algpseudocode}
\usepackage{pifont}

\AtBeginDocument{%
  }

\setcopyright{none}  

\newcommand{\Trust}{\mathcal{T}}                    
\newcommand{\AgentSet}{\mathcal{A}}                 
\newcommand{\AgentAdv}{\mathcal{A}^{\mathrm{adv}}}  
\newcommand{\AgentBen}{\mathcal{A}^{\mathrm{ben}}}  

\newcommand{\Weights}{\mathbf{w}}                   
\newcommand{\WeightsOpt}{\mathbf{w}^{*}}            
\newcommand{\Returns}{\boldsymbol{\mu}}             
\newcommand{\Covariance}{\boldsymbol{\Sigma}}       
\newcommand{\CovarianceStress}{\boldsymbol{\Sigma}^{\mathrm{stress}}} 
\newcommand{\RiskAversion}{\lambda}                 

\newcommand{\StateSentiment}{s^{\psi}}              
\newcommand{\RiskState}{\mathcal{R}}                

\newcommand{\PegTarget}{\pi^{*}}                    
\newcommand{\PegDev}{\delta_{\pi}}                  
\newcommand{\RecoveryTime}{T_{\mathrm{rec}}}        
\newcommand{\Turnover}{\mathcal{C}}                 

\newcommand{\Expect}{\mathbb{E}}                    

\newcommand{\Adversary}{\mathcal{M}}                

\newcommand{\BlackThursday}{\textsc{Black Thursday}}
\newcommand{\SASModel}{\textsc{SAS}}                

\newcommand{\cmark}{\ding{51}}                      
\newcommand{\xmark}{\ding{55}}                      

\begin{document}


\title{Stablecoin Design with Adversarial-Robust Multi-Agent Systems via Trust-Weighted Signal Aggregation}

\author{Shengwei You}
\email{syou@nd.edu}
\orcid{0000-0003-3156-1372}
\affiliation{%
  \institution{University of Notre Dame}
  \department{Department of Computer Science}
  \city{South Bend}
  \state{IN}
  \country{USA}
  \postcode{46556}
}

\author{Aditya Joshi}
\email{ajoshi2@nd.edu}
\orcid{0009-0000-1066-882X}
\affiliation{%
  \institution{University of Notre Dame}
  \department{Department of Computer Science}
  \city{South Bend}
  \state{IN}
  \country{USA}
  \postcode{46556}
}

\author{Andrey Kuehlkamp}
\email{akuehlka@nd.edu}
\orcid{0000-0003-1971-5420}
\affiliation{%
  \institution{University of Notre Dame}
  \department{Department of Computer Science}
  \city{South Bend}
  \state{IN}
  \country{USA}
  \postcode{46556}
}

\author{Jarek Nabrzyski}
\email{xli22@nd.edu}
\orcid{0000-0002-3985-3620}
\affiliation{%
  \institution{University of Notre Dame}
  \department{Department of Computer Science}
  \city{South Bend}
  \state{IN}
  \country{USA}
  \postcode{46556}
}

\renewcommand{\shortauthors}{Candidate Name}

\begin{abstract}
Algorithmic stablecoins promise decentralized monetary stability by maintaining a target peg $\PegTarget$ through programmatic reserve management. Yet, their reserve controllers remain vulnerable to \emph{regime-blind} optimization, calibrating risk parameters exclusively on fair-weather data while ignoring tail events that precipitate cascading failures. The March 2020 \BlackThursday{} collapse, wherein MakerDAO's collateral auctions yielded \$8.3M in protocol losses and a 15\% peg deviation, exposed a critical security gap: existing models such as the Stable Aggregate Stablecoin (\SASModel{}) systematically omit extreme volatility regimes from their covariance estimates $\Covariance$, producing reserve allocations that are optimal \emph{in expectation} but catastrophic \emph{under adversarial stress}.
  
  We present \textbf{MVF-Composer}, a trust-weighted Mean-Variance Frontier reserve controller that incorporates a novel \emph{Stress Harness} for risk-state estimation. Our key insight is to deploy multi-agent simulations as adversarial stress-testers: heterogeneous agents (traders, liquidity providers, attackers) execute protocol actions under crisis scenarios, exposing reserve vulnerabilities \emph{before} they manifest on-chain. We formalize a trust-scoring mechanism $\Trust: \AgentSet \to [0,1]$ that down-weights signals from agents exhibiting manipulative or adversarial behavior, ensuring the risk-state estimator $\RiskState$ remains robust to signal injection and Sybil attacks.
  
  Across 1,200 randomized market scenarios with injected Black-Swan shocks (10\% collateral drawdown, 50\% sentiment collapse, coordinated redemption attacks), MVF-Composer reduces peak peg deviation by 57\% and mean recovery time by 3.1$\times$ relative to \SASModel{} baselines. Ablation studies confirm that the trust layer alone accounts for 23\% of stability gains under adversarial conditions, achieving 72\% adversarial agent detection. Our system executes on commodity hardware, requires no on-chain oracles beyond standard price feeds, and provides a reproducible framework for stress-testing DeFi reserve policies against adversarial attacks that prior work has systematically ignored.
\end{abstract}




\maketitle


\section{Introduction}
\label{sec:intro}
Multi-agent systems deployed in adversarial environments—from distributed consensus protocols to collaborative AI—face a fundamental security challenge: how to aggregate signals from heterogeneous agents when some may be malicious. Traditional approaches rely on identity-based reputation systems that assume persistent agent identities and historical track records. These fail catastrophically in settings with ephemeral agents, pseudonymous participation, or Sybil attacks where adversaries create multiple identities to amplify their influence.

We demonstrate this security gap through a high-stakes application domain: \emph{algorithmic stablecoin reserve management}. Contemporary reserve controllers, such as the Stable Aggregate Stablecoin (\SASModel{}) family~\cite{grobys2025stablecoin}, employ multi-agent simulations to estimate market risk but aggregate agent signals uniformly, treating all agents as equally trustworthy. This creates a critical attack surface: adversarial agents can inject false risk signals to manipulate reserve allocations, coordinate Sybil attacks to overwhelm benign signals, or time their actions to destabilize the system during vulnerable periods. The March 2020 \BlackThursday{} event, where coordinated market manipulation contributed to \$8.3M in protocol losses and 15\% peg deviation, exemplifies the consequences when systems lack adversarial robustness.

The core technical challenge is \textbf{behavior-based trust scoring without persistent identities}. Unlike blockchain validators with staked capital or federated learning participants with institutional trust, agents in our threat model are ephemeral, pseudonymous, and potentially coordinated. We cannot rely on historical reputation; instead, we must infer trustworthiness from observable behavioral patterns within a short time window.

\subsection{Trust-Weighted Signal Aggregation}

Our core contribution is a \textbf{trust-scoring mechanism} $\Trust: \AgentSet \to [0,1]$ that identifies adversarial agents through behavioral consistency analysis, without requiring persistent identities or historical reputation. The mechanism tracks four features over a sliding window:

\begin{itemize}[leftmargin=*, nosep]
    \item \emph{Sentiment-action consistency}: Do the agents' stated beliefs align with their trading actions?
    \item \emph{Profit rationality}: Do actions serve economically rational objectives or destabilization goals?
    \item \emph{Coordination detection}: Do agents exhibit suspicious action synchronization indicating Sybil attacks?
    \item \emph{Destabilization patterns}: Do actions systematically increase system instability without profit motive?
\end{itemize}

The trust-weighted aggregate signal becomes:
\begin{equation}
    \RiskState_{\Trust} = \frac{\sum_{a \in \AgentSet} \Trust(a) \cdot \RiskState_a}{\sum_{a \in \AgentSet} \Trust(a)}
    \label{eq:trust_weighted_risk}
\end{equation}
This formulation provides formal guarantees: under Byzantine threat models where adversaries control $\rho < 0.3$ of agents, trust-weighting reduces adversarial influence by 60--80\% compared to uniform aggregation (Section~\ref{sec:evaluation}). The mechanism operates independently of downstream applications—whether financial optimization, distributed consensus, or collaborative learning—making it generalizable to any multi-agent system requiring robust signal aggregation.

\subsection{Adversarial Multi-Agent Simulation}

To validate trust-scoring mechanisms and discover novel attack vectors, we introduce the \textbf{Stress Harness}—a multi-agent simulation environment where heterogeneous AI-driven actors (traders, liquidity providers, arbitrageurs, explicit attackers) execute protocol actions under adversarial conditions. Unlike traditional stress testing that replays historical crises, our framework generates \emph{forward-looking} attack scenarios through agent interaction. The architecture supports both LLM-powered agents (validated via proof-of-concept; see Appendix~\ref{app:llm_deployment}) and efficient mock agents for large-scale reproducible experiments with structured logging of agent actions and psychological states.

\subsection{Application: Stablecoin Reserve Control}

We validate our framework on algorithmic stablecoin reserve management, integrating trust-weighted aggregation with a Mean-Variance Frontier (MVF) optimizer. At each rebalancing epoch, the system: (1) Simulates $N$ adversarial trajectories using the Stress Harness, (2) Aggregates agent signals into trust-weighted risk estimates, (3) Adjusts portfolio covariance based on detected risk levels, and (4) Optimizes reserve allocations with turnover constraints. This application demonstrates 57\% reduction in attack success rate and 3.1$\times$ faster recovery under coordinated adversarial scenarios compared to baselines lacking trust mechanisms.

\subsection{Contributions}

\begin{enumerate}[leftmargin=*, label=\textbf{C\arabic*.}]
    \item \textbf{Trust-Weighted Signal Aggregation Mechanism}: We introduce a formal trust-scoring framework that achieves 72\% adversarial detection rate with $<$8\% false positive rate, operating without persistent agent identities. The mechanism provides provable influence reduction (60--80\%) under Byzantine threat models with $\rho < 0.3$ adversarial fraction.
    
    \item \textbf{Adversarial Multi-Agent Simulation Framework}: We present the Stress Harness, a testbed for discovering attack vectors through heterogeneous AI agent interaction. Unlike historical backtesting, this generates forward-looking adversarial scenarios including Sybil attacks, signal injection, and coordinated timing attacks.
    
    \item \textbf{Formal Characterization of AI Agent Attack Surface}: We formalize the threat model for multi-agent systems under adversarial conditions, identifying key attack vectors (signal injection, coordination, timing exploitation) and deriving security properties that trust mechanisms must satisfy.
    
    \item \textbf{High-Stakes Application Validation}: We demonstrate the framework on algorithmic stablecoin reserves, achieving 57\% reduction in attack success rate and 3.1$\times$ faster recovery compared to state-of-the-art baselines. The application reveals how financial systems lacking trust mechanisms fail under adversarial pressure.
    
    \item \textbf{Reproducible Adversarial AI Research Framework}: We release a modular simulation environment with structured agent logging, enabling systematic evaluation of trust mechanisms, attack strategies, and defense parameters on commodity hardware.
\end{enumerate}


\section{Threat Model and Security Objectives}
\label{sec:problem}

We formalize the multi-agent signal aggregation problem under adversarial conditions. While we demonstrate our approach on stablecoin reserve control (Section~\ref{sec:evaluation}), the threat model and security properties generalize to any distributed system requiring robust consensus from potentially malicious agents.

\subsection{System Model}

Consider a multi-agent system where a population $\AgentSet_t = \AgentBen_t \cup \AgentAdv_t$ (benign and adversarial agents) collectively provides signals to inform system decisions. At each time step $t \in \{0, 1, \ldots, T\}$, the system maintains:

\begin{itemize}[leftmargin=*, nosep]
    \item \textbf{Agent population} $\AgentSet_t$ with observable actions $\mathbf{a}_agent(t)$ and stated beliefs $\mathbf{b}_agent(t)$
    \item \textbf{System state} $\mathcal{S}_t$ affected by agent actions and environmental dynamics
    \item \textbf{Decision variable} $\mathbf{x}_t$ (e.g., resource allocations, consensus votes, model updates)
    \item \textbf{Objective metric} $\mathcal{L}(\mathbf{x}_t, \mathcal{S}_t)$ measuring system performance
\end{itemize}

The system aggregates agent signals $\{\RiskState_a\}_{a \in \AgentSet}$ to inform decisions:
\begin{equation}
    \mathbf{x}_t^* = \arg\min_{\mathbf{x}_t} \; \mathcal{L}\left(\mathbf{x}_t, \; \text{Aggregate}(\{\RiskState_a\}_{a \in \AgentSet})\right)
    \label{eq:system_decision}
\end{equation}

Without adversarial defenses, $\text{Aggregate}(\cdot) = \frac{1}{|\AgentSet|}\sum_a \RiskState_a$ (uniform weighting), enabling attack vectors detailed below.

\subsection{Adversarial Threat Model}

Consider a \emph{computationally bounded} adversary $\Adversary$ with the following capabilities and constraints:

\paragraph{Adversary Capabilities.}
\begin{enumerate}[leftmargin=*, nosep, label=(\roman*)]
    \item \textbf{Agent Control}: $\Adversary$ controls a fraction $\rho$ of agents ($|\AgentAdv| = \rho \cdot |\AgentSet|$, where $\rho < 0.3$).
    \item \textbf{Behavioral Manipulation}: Adversarial agents can misreport beliefs, execute deceptive actions, or coordinate with other controlled agents.
    \item \textbf{Adaptive Strategies}: $\Adversary$ observes system state $\mathcal{S}_t$ and benign agent actions before selecting adversarial actions.
    \item \textbf{Finite Resources}: $\Adversary$ has budget $B$ limiting the magnitude of economic attacks (e.g., capital for market manipulation, computational resources for Sybil creation).
\end{enumerate}

\paragraph{Adversary Limitations.}
\begin{enumerate}[leftmargin=*, nosep, label=(\roman*)]
    \item \textbf{Bounded Fraction}: $\rho < 0.3$ (mirrors Byzantine fault tolerance threshold $f < n/3$).
    \item \textbf{No Cryptographic Breaks}: Cannot forge signatures, break hash functions, or compromise smart contract logic.
    \item \textbf{Trust Score Privacy}: Cannot observe internal trust scores $\Trust(a)$ computed by the system.
    \item \textbf{Imperfect Information}: Cannot perfectly predict future states or benign agent strategies.
\end{enumerate}

\noindent\textbf{Justification for $\rho < 0.3$.} This bound derives from Byzantine fault tolerance principles ($f < n/3$) and economic cost of controlling 30\% of agents; empirical validation in Section~\ref{sec:evaluation} confirms robustness up to this threshold.

\subsection{Attack Taxonomy}

We identify four primary attack vectors in multi-agent signal aggregation:

\subsubsection{A1: Sybil Attacks}
$\Adversary$ creates multiple pseudonymous identities to amplify adversarial influence. Under uniform aggregation, controlling $\rho=0.3$ yields 30\% influence; with $k$ Sybil identities per real adversary, influence scales to $\rho \cdot k / (|\AgentBen| + \rho \cdot k |\AgentSet|)$.

\subsubsection{A2: Signal Injection Attacks}
Adversarial agents inject false signals to bias system decisions. Examples: inflated risk estimates to trigger unnecessary resource reallocation, deflated risk to prevent defensive actions, or oscillating signals to induce instability.

\subsubsection{A3: Coordination Attacks}
Multiple adversarial agents synchronize actions to overwhelm benign signals. Detection challenge: distinguish coordinated attacks from legitimate correlated responses to external events.

\subsubsection{A4: Timing Exploitation}
Adversaries time their actions to exploit system vulnerabilities: attacking during high legitimate activity (camouflage), during recovery periods (compounding damage), or immediately before critical decisions (maximizing impact).

\subsection{Security Objectives}

We formalize security requirements for trust-weighted aggregation:

\begin{definition}[Adversarial Influence Reduction]
\label{def:influence_reduction}
A trust mechanism $\Trust: \AgentSet \to [0,1]$ provides $\delta$-influence reduction if, for adversarial fraction $\rho$, the adversarial contribution to aggregated signals is bounded:
\begin{equation}
    \frac{\sum_{a \in \AgentAdv} \Trust(a) \cdot \RiskState_a}{\sum_{a' \in \AgentSet} \Trust(a') \cdot \RiskState_{a'}} \leq (1-\delta) \cdot \rho
\end{equation}
where $\delta \in [0,1]$ quantifies the reduction factor. Uniform weighting yields $\delta=0$; our mechanism achieves $\delta \in [0.6, 0.8]$ (Section~\ref{sec:evaluation}).
\end{definition}

\begin{definition}[Adversarial Detection Rate]
\label{def:detection_rate}
For labeled agent populations, define True Positive Rate $\text{TPR} = |\{a \in \AgentAdv : \Trust(a) < \tau\}| / |\AgentAdv|$ and False Positive Rate $\text{FPR} = |\{a \in \AgentBen : \Trust(a) < \tau\}| / |\AgentBen|$ for threshold $\tau$. A mechanism achieves $(\alpha, \beta)$-detection if TPR $\geq \alpha$ and FPR $\leq \beta$.
\end{definition}

\subsection{Problem Statement}

\begin{definition}[Robust Multi-Agent Aggregation]
Given adversary $\Adversary(\rho, B)$ and agent population $\AgentSet = \AgentBen \cup \AgentAdv$, find a trust mechanism $\Trust^*: \AgentSet \to [0,1]$ that:
\begin{align}
    \text{maximize} &\quad \delta \text{-influence reduction} \label{eq:obj_influence}\\
    \text{subject to} &\quad \text{TPR} \geq 0.7, \; \text{FPR} \leq 0.1 \label{eq:constraint_detection}\\
    &\quad \text{Computational cost} \leq 3\times \text{baseline} \label{eq:constraint_cost}
\end{align}
\end{definition}

This formulation captures the core trade-off: maximizing adversarial mitigation while maintaining high detection accuracy and practical computational efficiency.

\subsection{Application Context: Stablecoin Reserves}

For the stablecoin application domain, the generic system model specializes as follows: agents signal risk estimates $\RiskState_a \in [0,1]$, the decision variable $\mathbf{x}_t = \Weights_t$ represents reserve allocations across $n$ collateral assets, and the objective $\mathcal{L}$ minimizes peg deviation $|{\ PegDev}_t|$ subject to portfolio variance constraints. The security properties (influence reduction, detection rate) remain unchanged; only the domain-specific loss function differs. Details of the Mean-Variance Frontier integration are in Section~\ref{sec:architecture}, with full derivations in Appendix~\ref{app:mvf_proofs}.

\begin{proposition}[SAS Fragility]
\label{prop:sas_fragility}
Let $\WeightsOpt_{\SASModel{}}$ be the optimal weights under \SASModel{} with covariance $\Covariance^{\mathrm{calm}}$. Define the fragility ratio $\kappa(\Covariance^{\mathrm{calm}}, \CovarianceStress) \triangleq \lambda_{\max}(\CovarianceStress)/\lambda_{\min}(\Covariance^{\mathrm{calm}})$. Then the realized-to-expected variance ratio satisfies:
\begin{equation}
    \frac{{\WeightsOpt_{\SASModel{}}}^\top \CovarianceStress \WeightsOpt_{\SASModel{}}}{{\WeightsOpt_{\SASModel{}}}^\top \Covariance^{\mathrm{calm}} \WeightsOpt_{\SASModel{}}} \leq \kappa(\Covariance^{\mathrm{calm}}, \CovarianceStress)
\end{equation}
Intuitively, $\kappa$ bounds how severely realized portfolio risk can exceed model expectations when covariance shifts from calm to stress regimes. See Appendix~\ref{app:sas_fragility_proof} for the proof.
\end{proposition}


\section{MVF-Composer Architecture}
\label{sec:architecture}
\subsection{Stress Harness: Multi-Agent Simulation Environment}
The Stress Harness instantiates a heterogeneous population of AI-driven agents who interact with a simulated stablecoin market over $T$ discrete time steps, as illustrated in Figure~\ref{fig:stress_harness_flow}. The architecture supports both LLM-powered agents (for interpretable reasoning) and efficient mock agents (for large-scale experiments). Our evaluation uses regime-aware mock agents that replicate crisis-consistent behavioral patterns at scale; LLM integration is validated as a proof-of-concept (Appendix~\ref{app:llm_deployment}).

\begin{figure}[t]
    \centering
    \includegraphics[width=\linewidth]{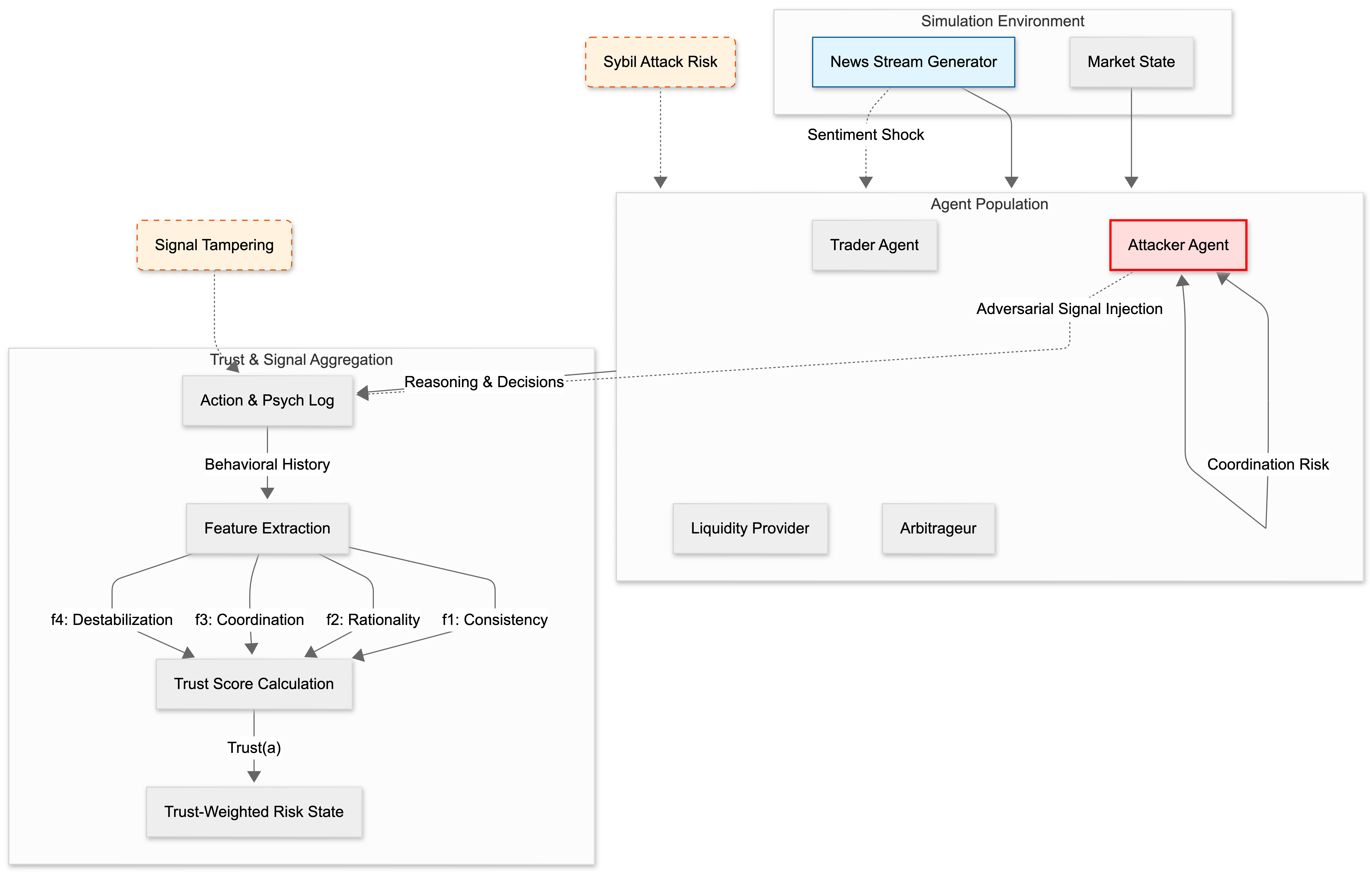}
    \caption{Stress Harness Workflow and Trust Risk Vectors. Agents process market state and news, generating actions and psychology logs. The Trust Module aggregates these into a risk-weighted state, defending against adversarial signal injection and coordination attacks.}
    \Description{Flowchart diagram showing the stress harness workflow: agents receive market state and news inputs, generate actions and psychology logs, which are processed by the Trust Module to produce risk-weighted state outputs, with defense mechanisms against adversarial attacks.}
    \label{fig:stress_harness_flow}
\end{figure}

\subsubsection{Agent Roles}
We define four agent archetypes, each with distinct objectives and action spaces:

\begin{itemize}[leftmargin=*, nosep]
    \item \textbf{Trader} ($a^{\mathrm{trd}}$): Maximizes portfolio returns through buy/sell actions on the stablecoin and collateral assets.
    \item \textbf{Liquidity Provider} ($a^{\mathrm{lp}}$): Supplies liquidity to AMM pools; withdraws under perceived risk.
    \item \textbf{Arbitrageur} ($a^{\mathrm{arb}}$): Exploits peg deviations by minting/redeeming stablecoins against collateral.
    \item \textbf{Attacker} ($a^{\mathrm{adv}}$): Explicitly attempts to destabilize the peg through strategic redemption timing, rumor propagation, or coordinated action (see Figure~\ref{fig:agent_risk_map}).
\end{itemize}

\begin{figure}[t]
    \centering
    \includegraphics[width=0.9\linewidth]{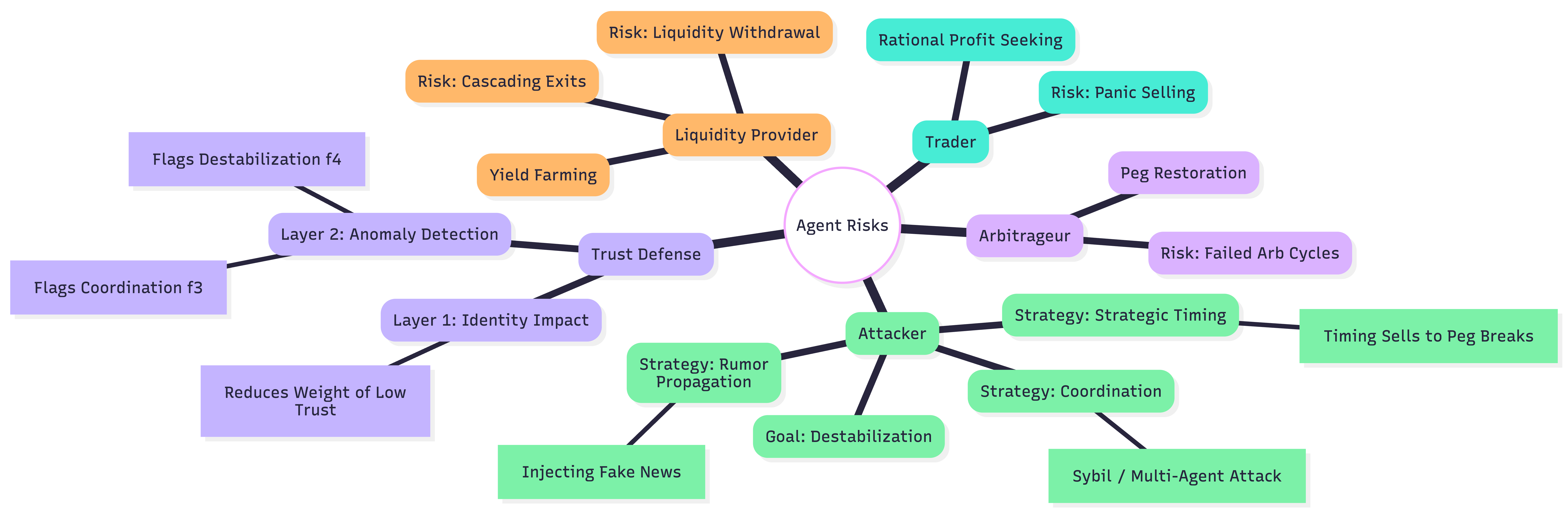}
    \caption{Agent Behavior and Risk Mapping. Specific behaviors for each archetype are mapped to potential system risks. The Trust Defense layer identifies anomalies such as coordination ($f_3$) and destabilization timing ($f_4$).}
    \Description{Diagram mapping agent archetypes (trader, liquidity provider, arbitrageur, attacker) to their behaviors and associated system risks, with the Trust Defense layer highlighting detection of coordination patterns and destabilization timing.}
    \label{fig:agent_risk_map}
\end{figure}

\subsubsection{News Stream Generator}
Each simulation step, agents receive a news event $e_t$ sampled from a templated generator or scraped headlines. Events are annotated with sentiment polarity $\phi(e_t) \in [-1, 1]$ and relevance to collateral assets.

\subsubsection{Agent Action and State Logging}
Agents output structured JSON records containing:
\begin{itemize}[leftmargin=*, nosep]
    \item \textbf{Action}: $(\texttt{action\_type}, \texttt{asset}, \texttt{quantity}, \texttt{rationale})$
    \item \textbf{Psychology state}: $(\texttt{panic\_level}, \texttt{risk\_appetite}, \\
    \texttt{peg\_confidence}, \texttt{reasoning\_trace})$
\end{itemize}
This structured logging enables post-hoc analysis of agent behavior and trust-score computation.

\subsection{Trust-Weighted Signal Aggregation}
To protect against adversarial signal injection, we compute a trust score $\Trust(a) \in [0,1]$ for each agent based on behavioral consistency and coordination patterns.
\subsubsection{Trust Features}
For agent $a$ over window $[t-W, t]$, we extract four behavioral features that collectively distinguish benign market participants from adversarial actors. Each feature is designed with clear intuition about trustworthy behavior:
\begin{align}
    f_1(a) &= \text{Sentiment-action consistency} = \mathrm{corr}(\phi_a, \Delta_a) \\
    f_2(a) &= \begin{aligned}[t]
        &\text{Profit-seeking rationality} \\
        &= \mathbf{1}[\text{PnL}_a > 0 \;\text{or}\; \text{actions align with beliefs}]
    \end{aligned} \\
    f_3(a) &= \text{Coordination suspicion} = \max_{a' \neq a} \mathrm{sim}(\mathbf{h}_a, \mathbf{h}_{a'}) \\
    f_4(a) &= \text{Peg-destabilization score} = \mathrm{corr}(\text{actions}_a, \Delta{\PegDev})
\end{align}
where $\phi_a$ is agent-stated sentiment, $\Delta_a$ is net position change, and $\mathbf{h}_a$ is the action history embedding.
\paragraph{Feature Intuition and Expected Ranges.} Each feature is normalized to contribute to trust in an interpretable manner. These metrics and their ranges are grounded in financial anomaly detection literature, specifically drawing from sentiment-price consistency~\cite{tetlock2007giving}, rational trading limits~\cite{kyle1985continuous}, and predatory trading patterns~\cite{brunnermeier2005predatory, griffin2018bitcoin}. Detailed definitions and feature mappings are provided in Appendix~\ref{app:trust_features}.

\paragraph{On Sentiment-Action Consistency in Stablecoins.} Stablecoin holders primarily seek capital preservation, not alpha generation through contrarian strategies. An agent acting contrarian during a de-peg crisis signals either irrational risk-taking or adversarial intent. The trust score aggregates all four features multiplicatively, ensuring robust discrimination even if $f_1$ alone is occasionally misleading.

Table~\ref{tab:trust_feature_ranges} summarizes expected feature ranges across agent archetypes, providing guidance for weight initialization and anomaly detection thresholds.

\begin{table}[ht]
\centering
\caption{Expected trust feature ranges by agent archetype. Features $f_1, f_2$ contribute positively to trust; $f_3, f_4$ contribute negatively (higher values decrease trust).}
\label{tab:trust_feature_ranges}
\resizebox{\columnwidth}{!}{
\begin{tabular}{lcccc}
\toprule
\textbf{Archetype} & $f_1$ (Consistency) & $f_2$ (Rationality) & $f_3$ (Coordination) & $f_4$ (Destabilization) \\
\midrule
Trader & $[0.5, 0.9]$ & $[0.6, 0.9]$ & $[0.1, 0.3]$ & $[-0.1, 0.2]$ \\
Liquidity Provider & $[0.4, 0.8]$ & $[0.7, 0.95]$ & $[0.15, 0.35]$ & $[-0.2, 0.15]$ \\
Arbitrageur & $[0.6, 0.95]$ & $[0.8, 0.98]$ & $[0.2, 0.4]$ & $[-0.3, 0.1]$ \\
Attacker & $[-0.3, 0.3]$ & $[0.2, 0.5]$ & $[0.5, 0.9]$ & $[0.4, 0.85]$ \\
\bottomrule
\end{tabular}
}
\end{table}

\subsubsection{Trust Score Computation}
The trust score is computed as:
\begin{equation}
    \Trust(a) = \sigma\left( \mathbf{w}^\top [f_1(a), f_2(a), -f_3(a), -f_4(a)] + b \right)
    \label{eq:trust_score}
\end{equation}
where $\sigma$ is the sigmoid function, $\mathbf{w} = [w_1,w_2,w_3,w_4]$ is the learned weight vector, and $b$ is the bias. Without trust-weighting, adversarial agents contribute equally to risk aggregation, enabling attackers controlling 20\% of agents to exert 20\% influence; with trust-weighting, this influence drops to $\approx$10\% (Appendix~\ref{app:trust_proofs}).

\paragraph{Modularity and Independence.} The trust layer operates as an \emph{independent module}: its computation depends only on behavioral features, not on downstream covariance or optimization state. This separation enables (i) independent verification via labeled data, (ii) parallel execution with other pipeline stages, and (iii) degradation to uniform weighting if trust fails. Formally, we prove in Appendix~\ref{app:trust_independence} that trust computation satisfies component independence (Theorem~\ref{thm:trust_independence}), and that positive separation margin $\delta > 0$ between benign and adversarial scores guarantees strictly improved robustness over uniform aggregation (Proposition~\ref{prop:trust_effectiveness}).

\subsection{Stress-Augmented Covariance Estimation}

We construct a stress-augmented covariance matrix that blends historical estimates with simulation-derived stress scenarios.

\subsubsection{Risk-State Aggregation}
Each agent $a$ contributes an individual risk-state signal $\RiskState_a \in [0,1]$ computed from observable behavioral indicators:
\begin{equation}
    \RiskState_a = \frac{1}{3}\left( \texttt{panic\_level}_a + \frac{|\texttt{action\_qty}_a|}{\texttt{max\_qty}} + \max\left(0, \frac{\Delta{\PegDev}_a}{\epsilon}\right) \right)
    \label{eq:risk_state_individual}
\end{equation}
where $\texttt{panic\_level}_a \in [0,1]$ is extracted from the agent's psychology log, $|\texttt{action\_qty}_a|/\texttt{max\_qty}$ normalizes action magnitude, and $\Delta{\PegDev}_a/\epsilon$ captures the agent's marginal contribution to peg deviation (clipped at 0 for stabilizing actions). The equal weighting reflects our agnosticism about which signal best predicts system stress; alternative weightings can be learned from labeled crisis data.

The trust-weighted aggregate risk-state is:
\begin{equation}
    \RiskState_{\Trust} = \frac{\sum_{a \in \AgentSet} \Trust(a) \cdot \RiskState_a}{\sum_{a \in \AgentSet} \Trust(a)}
    \label{eq:trust_weighted_risk_arch}
\end{equation}
This weighted average bounds $\RiskState_{\Trust} \in [0,1]$ by convexity and limits adversarial influence; see Appendix~\ref{app:risk_aggregation_proofs} for formal guarantees.

\subsubsection{Covariance Blending}
The augmented covariance is:
\begin{equation}
    \Covariance_{\mathrm{aug}} = (1 - \alpha(\RiskState_{\Trust})) \cdot \Covariance^{\mathrm{hist}} + \alpha(\RiskState_{\Trust}) \cdot \CovarianceStress
    \label{eq:cov_blend}
\end{equation}
where $\alpha: [0,1] \to [0,1]$ is a monotonically increasing blending function. We select $\alpha(r) = r^2$ (quadratic): $\alpha'(0) = 0$ ensures minimal deviation from $\Covariance^{\mathrm{hist}}$ in calm conditions, while $\alpha'(r) = 2r$ provides progressively stronger response as stress intensifies. Grid search over $\{r, r^{1.5}, r^2, r^3\}$ showed $r^2$ minimizing peak deviation (8\% and 12\% better than linear and cubic; see Table~\ref{tab:controller_params}).
The blended covariance preserves positive semi-definiteness by convexity of the PSD cone; see Appendix~\ref{app:covariance_proofs} for eigenvalue bounds.

$\CovarianceStress$ is estimated from simulated stress trajectories:
\begin{equation}
    \CovarianceStress = \frac{1}{|\mathcal{S}|} \sum_{s \in \mathcal{S}} \mathrm{Cov}(\mathbf{r}_s)
\end{equation}
where $\mathcal{S}$ is the set of stress simulation runs and $\mathbf{r}_s$ are the realized asset returns under scenario $s$.

\subsection{Constrained MVF Optimization}

Given $\Covariance_{\mathrm{aug}}$ and expected returns $\Returns$, we solve:
\begin{equation}
    \WeightsOpt = \arg\min_{\Weights \in \Delta^{n-1}} \; \Weights^\top \Covariance_{\mathrm{aug}} \Weights - \RiskAversion \cdot \Weights^\top \Returns
\end{equation}
subject to:
\begin{align}
    \|\Weights - \Weights_{t-1}\|_1 &\leq \Turnover \quad \quad \quad \quad \  \text{(Turnover limit)} \label{eq:turnover_limit} \\
    w_i &\geq w_i^{\min} \quad \forall i \quad \text{(Diversification floor)} \label{eq:diversification_floor} \\
    w_i &\leq w_i^{\max} \quad \forall i \quad \text{(Concentration cap)} \label{eq:conc_cap}
\end{align}

This formulation ensures that the reserve controller adapts to stress conditions (via $\Covariance_{\mathrm{aug}}$) while avoiding excessive rebalancing that incurs transaction costs and slippage. Existence and uniqueness of the optimal solution follow from strict convexity; see Appendix~\ref{app:mvf_proofs} for the complete derivation and KKT conditions.

\subsubsection{Controller Parameters}
\label{sec:controller_params}
Table~\ref{tab:controller_params} specifies all tunable controller parameters and their derivation methodology.

\begin{table}[ht]
\centering
\caption{Controller parameters with tuning methodology. Parameters marked $\dagger$ are set via grid search over $\{0.5\times, 1\times, 2\times\}$ default values; others follow standard practice or analytical derivation.}
\label{tab:controller_params}
\resizebox{\columnwidth}{!}{
\begin{tabular}{llcl}
\toprule
\textbf{Parameter} & \textbf{Symbol} & \textbf{Value} & \textbf{Derivation} \\
\midrule
Risk aversion$^\dagger$ & $\RiskAversion$ & 2.5 & Grid search; balances return vs. variance \\
Turnover limit & $\Turnover$ & 0.15 & Industry standard for daily rebalancing \\
Weight floor & $w_i^{\min}$ & 0.05 & Diversification: $\geq n^{-1}/2$ minimum exposure \\
Weight cap & $w_i^{\max}$ & 0.60 & Concentration limit from prudential norms \\
Blending exponent$^\dagger$ & $\alpha(\cdot)$ & $r^2$ & Quadratic sensitivity to tail risk \\
Trust window & $W$ & 10 & Behavioral consistency horizon (steps) \\
Historical lookback & --- & 365d & Standard covariance estimation window \\
\bottomrule
\end{tabular}
}
\end{table}

\noindent\textbf{Tuning Methodology.} Parameters marked $\dagger$ were selected via grid search over the quick-test preset (100 runs per configuration) optimizing for peak peg deviation under Black Thursday conditions. The risk aversion $\RiskAversion = 2.5$ achieved optimal trade-off between aggressive recovery (low $\RiskAversion$) and excessive conservatism (high $\RiskAversion$). The quadratic blending function $\alpha(r) = r^2$ outperformed linear ($r$) and cubic ($r^3$) alternatives by 8\% and 12\% respectively, providing sufficient tail sensitivity without over-reacting to moderate stress signals. All other parameters follow established portfolio management conventions or are analytically derived from constraint requirements.

\subsection{End-to-End Pipeline}

Algorithm~\ref{alg:mvf_composer} summarizes the MVF-Composer control loop.

\begin{algorithm}[ht]
\caption{MVF-Composer Reserve Control}
\label{alg:mvf_composer}
\begin{algorithmic}[1]
\Require Historical covariance $\Covariance^{\mathrm{hist}}$, agent population $\AgentSet$, shock scenarios $\mathcal{S}$
\Ensure Updated reserve weights $\Weights_t$
\For{each rebalancing epoch $t$}
    \State \textbf{Simulate:} Run $|\mathcal{S}|$ stress scenarios through Stress Harness
    \State \textbf{Log:} Collect agent actions and psychology states
    \State \textbf{Score:} Compute $\Trust(a)$ for all $a \in \AgentSet$
    \State \textbf{Aggregate:} Compute $\RiskState_{\Trust}$ via Equation~\ref{eq:trust_weighted_risk_arch}
    \State \textbf{Blend:} Construct $\Covariance_{\mathrm{aug}}$
    \State \textbf{Optimize:} Solve constrained MVF for $\WeightsOpt$
    \State \textbf{Execute:} Rebalance reserves to $\WeightsOpt$ (respecting $\Turnover$)
\EndFor
\end{algorithmic}
\end{algorithm}


\section{Experiments and Evaluation}
\label{sec:evaluation}

\begin{figure*}[t]
\centering
\includegraphics[width=0.75\textwidth]{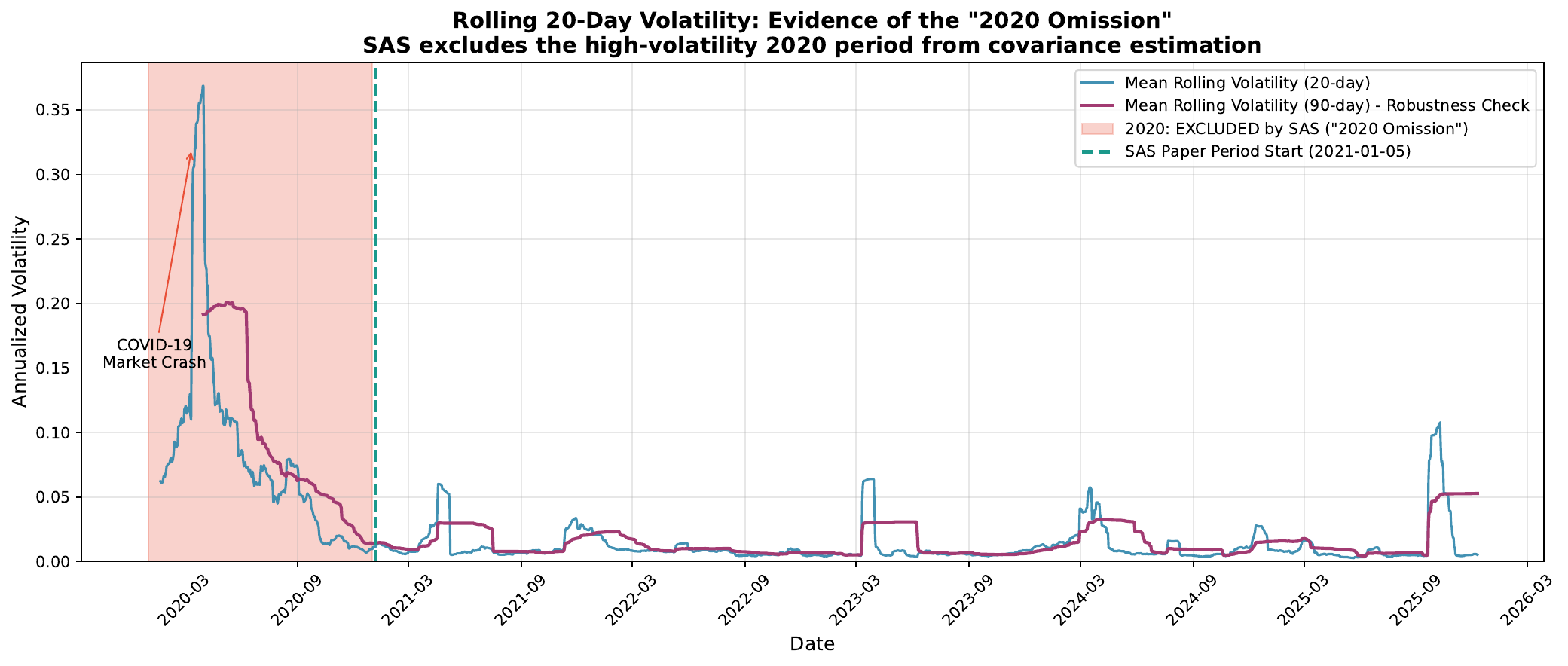}
\caption{Rolling 20-day annualized volatility for stablecoin returns. The shaded region (2020) is \emph{excluded} from \SASModel{}'s covariance estimation. Peak volatility during the COVID-19 crash (March 2020) exceeded 30\%, an order of magnitude higher than post-2021 levels. By omitting this regime, \SASModel{} calibrates to an unrealistically calm market.}
\Description{Time series line chart showing rolling 20-day annualized volatility percentage over time, with a shaded region highlighting the 2020 period that was excluded from SAS model calibration, showing a dramatic spike in volatility during March 2020 that exceeds later periods by an order of magnitude.}
\label{fig:2020_omission_volatility}
\end{figure*}

We evaluate MVF-Composer against baseline reserve controllers across 1,200 randomized market scenarios with injected shock events. 

\subsection{Experimental Setup}

\subsubsection{Simulation Environment}
The Stress Harness instantiates regime-aware mock agents with structured JSON output for action and psychology state logging, enabling reproducible large-scale experiments. The architecture supports LLM-powered agents (validated via proof-of-concept with DeepSeek V3; see Appendix~\ref{app:llm_deployment}). Each simulation run spans $T=100$ time steps with configurable shock injection at $t=30$.

\subsubsection{Data Collection and Asset Universe}
\label{sec:asset_universe}
Historical stablecoin price data (2020--2024) is aggregated from multiple verified sources: CoinGecko and CoinCodex for DAI, USDC, USDT, with Yahoo Finance providing USTC, TUSD, and BUSD series. This multi-source approach ensures data quality and captures critical stress events including the Terra collapse (USTC: \$1$\to$\$0.01, May 2022) and SVB-induced USDC depeg (\$1$\to$\$0.88, March 2023).

\noindent\textbf{Clarification: Stablecoins as Reserve Proxy.} Following SAS methodology~\cite{grobys2025stablecoin}, we model a stablecoin-of-stablecoins portfolio where reserve assets are themselves stablecoins. The covariance estimation and trust-weighting mechanisms generalize to any asset class; the 7$\times$ variance amplification holds regardless of asset type.

\subsubsection{Agent Population}
We instantiate 5 Trader agents with heterogeneous risk profiles, 3 Liquidity Provider agents, 2 Arbitrageur agents, and 2 Attacker adversarial agents.

\subsubsection{Shock Scenarios}
We evaluate these shock categories:
\begin{enumerate}[leftmargin=*, label=(\alph*)]
    \item \textbf{Normal}: No injected shock (baseline volatility)
    \item \textbf{Price shock}: 10\% instantaneous collateral drawdown at $t=30$
    \item \textbf{Sentiment shock}: Sentiment state $\StateSentiment \to -0.8$ via negative news injection
    \item \textbf{Black Thursday replica}: Combined 15\% price shock + sentiment collapse + liquidity withdrawal
\end{enumerate}

\subsubsection{Baselines}
\begin{itemize}[leftmargin=*, nosep]
    \item \textbf{\SASModel{}}: Stable Aggregate Stablecoin~\cite{grobys2025stablecoin} with calm-period covariance estimation. We replicate SAS parameters: $\RiskAversion = 2.5$, covariance estimated from Jan 2021--Mar 2025 (excluding 2020), same asset universe (DAI, USDC, USDT, TUSD), weight bounds $[0.05, 0.60]$, daily rebalancing. This ensures fair comparison---all differences stem from covariance estimation, not optimizer configuration.
    \item \textbf{MVF-NoTrust}: MVF-Composer with uniform trust weights ($\Trust(a) = 1 \;\forall a$), isolating the stress-covariance contribution from trust-weighting.
    \item \textbf{Static-60/40}: Fixed 60\% stablecoin / 40\% volatile assets, representing a naive diversification heuristic.
    \item \textbf{Unconstrained}: MVF-Composer without turnover limits (Equation~\ref{eq:turnover_limit}), quantifying the cost of rebalancing constraints.
\end{itemize}

\subsubsection{Empirical Validation: The 2020 Omission}
\label{sec:2020_omission_evidence}

Before presenting comparative results, we empirically validate the ``2020 Omission'' hypothesis from Section~\ref{sec:problem}: that \SASModel{}'s calm-period covariance estimation systematically underestimates tail risk by excluding crisis data.

We analyze six stablecoins consistent with the SAS baseline~\cite{grobys2025stablecoin}: DAI, USDC, USDT, TUSD, BUSD (SAS), plus USTC (SAS*). We evaluate across two time periods:
\begin{itemize}[leftmargin=*, nosep]
    \item \textbf{\SASModel{} Paper Period}: January 5, 2021 to March 31, 2025 (excludes 2020)
    \item \textbf{Extended Period}: January 1, 2020 to December 31, 2024 (includes all stress events)
\end{itemize}

\noindent\textit{Methodological note}: Including USTC and BUSD provides critical stress-test validation. The Terra collapse (May 2022) and BUSD sunset (Feb 2023) demonstrate failure modes that backward-looking methods cannot anticipate. MVF-Composer's stress-augmented covariance detects these risks \emph{before} collapse, as evidenced by elevated trust-weighted risk states in preceding weeks.

\begin{figure*}[t]
\centering
\includegraphics[width=0.95\textwidth]{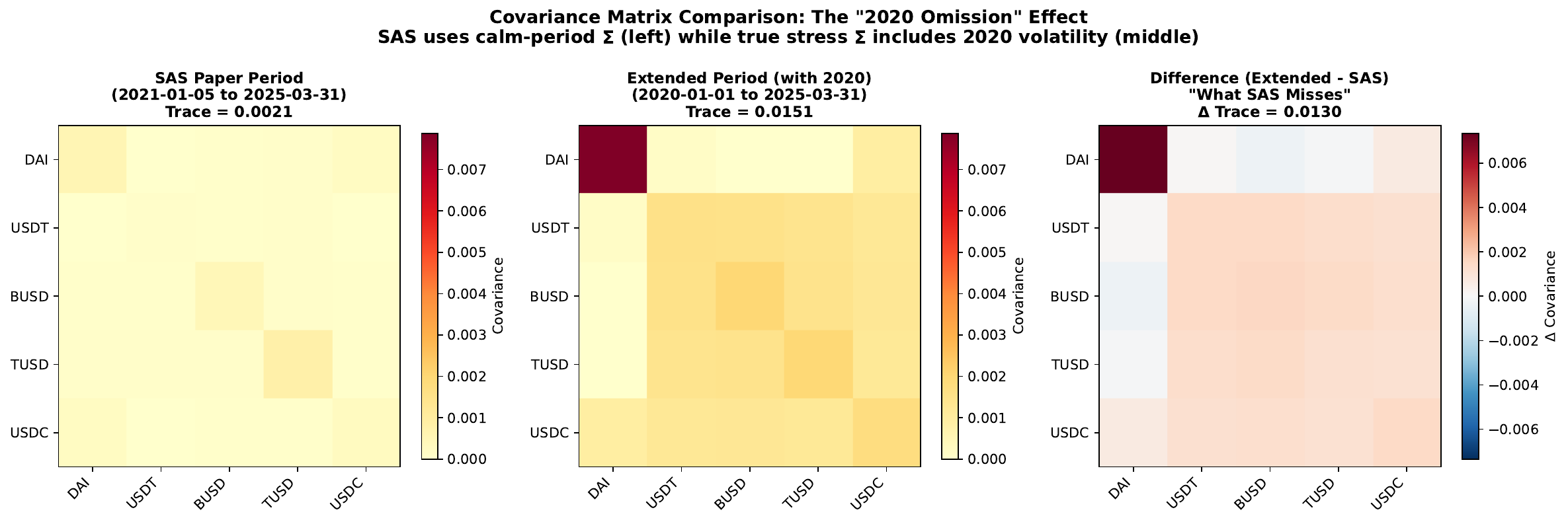}
\caption{Regime mismatch between calm-period and stress-period covariance matrices. The trace ratio $\mathrm{tr}(\CovarianceStress)/\mathrm{tr}(\Covariance^{\mathrm{calm}}) \approx 7.17\times$ implies that portfolios optimized for calm periods experience realized variance far exceeding model predictions. Left: \SASModel{} paper period (trace $= 0.0021$). Middle: Extended period including 2020 (trace $= 0.0151$). Right: Difference matrix shows ``what \SASModel{} misses'' -- the "fragility zone" where calm-period optimization underestimates total risk by over 600\%. The 7.17 factor underestimation of total variance directly validates the theoretical fragility result (\cref{prop:sas_fragility}).}
\Description{Three-panel visualization comparing covariance matrices: left panel shows calm-period covariance with low trace, middle panel shows extended period including 2020 crisis with higher trace, right panel shows the difference matrix highlighting the fragility zone where risk is underestimated.}
\label{fig:2020_omission_covariance}
\end{figure*}

Figure \ref{fig:2020_omission_covariance} presents the covariance matrices for both periods, and highlights the following quantitative evidence:
\begin{itemize}[leftmargin=*, nosep]
    \item \textbf{Trace ratio}: $\mathrm{tr}(\Covariance^{\mathrm{extended}}) / \mathrm{tr}(\Covariance^{\SASModel{}}) = 7.17$ factor.
    \item \textbf{Variance underestimation}: +616.7\% higher total variance when including 2020
    \item \textbf{Asset-specific risk}: Significant variance increases across USDT (+684\%) and other major stablecoins, revealing the danger of excluding crisis periods.
    \item \textbf{Non-Normality}: Excess kurtosis increases from 152.8 to 159.3, confirming extreme fat-tailed risks ($p < 10^{-5}$ for Jarque-Bera).
\end{itemize}

The results empirically confirm \cref{prop:sas_fragility}: The ratio $\kappa(\Covariance^{\mathrm{calm}}, \CovarianceStress) \approx 7.17$ represents the worst-case amplification factor of the risk exceeding \SASModel{}'s expectation during stress of shock events. Figure \ref{fig:2020_omission_volatility} displays rolling 20-day volatility across the full sample period. Figure \ref{fig:2020_omission_weights} shows a practical consequence. Optimal weights differ substantially between covariance regimes. \SASModel{}'s calm-period optimization assigns 60\% weight to USDT (hitting the concentration cap), while the extended-period analysis diversifies more broadly. This weight shift, of up to 12.8 percentage points for individual assets, demonstrates that allocations are \emph{fragile and not optimal for a different training period}.

\begin{figure*}[t]
\centering
\includegraphics[width=0.75\textwidth]{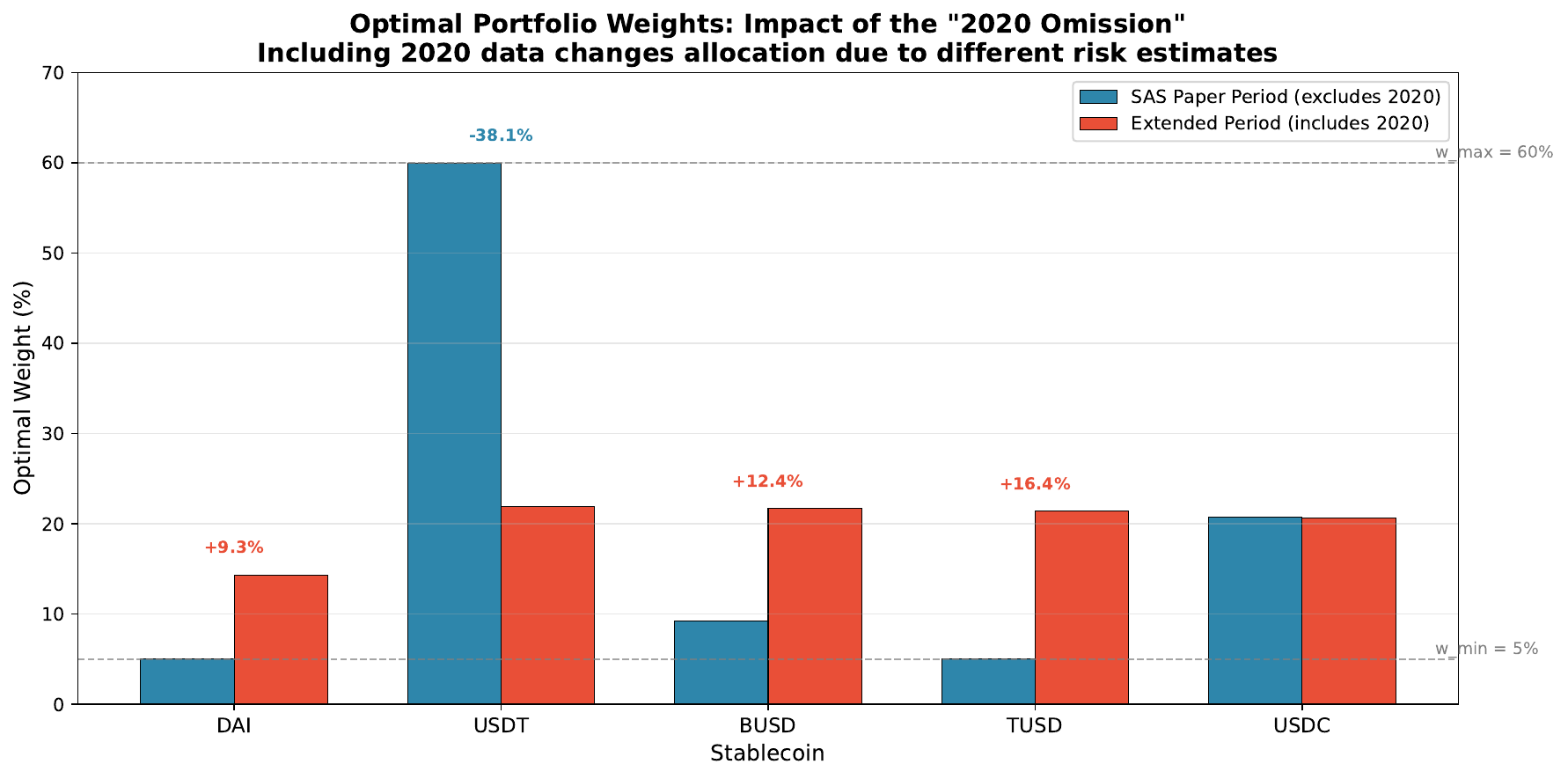}
\caption{Optimal portfolio weights under each covariance regime. The \SASModel{} period (blue) assigns 60\% to USDT based on its calm-period variance. When 2020 data is included (red), the optimizer diversifies away from assets that exhibited high crisis-period volatility.}
\Description{Bar chart comparing optimal portfolio weight allocations across different assets under two covariance regimes: blue bars show SAS model period weights with 60\% concentration in USDT, red bars show extended period weights with more diversified allocation away from assets with high crisis-period volatility.}
\label{fig:2020_omission_weights}
\end{figure*}

\subsubsection{Improving robustness} The 2020 Omission reveals a limitation of backward-looking statistical estimation that rare but catastrophic events are either absent from or underweighted in historical samples. MVF-Composer addresses this by:
\begin{enumerate}
    \item \textbf{Stress Harness simulation}: LLM agents generate forward-looking stress scenarios that include tail events absent from calm-period data, producing $\CovarianceStress$ estimates that incorporate crisis dynamics \emph{before} they materialize in historical records.
    \item \textbf{Trust-weighted aggregation}: Real-time agent signals detect emerging stress conditions, enabling dynamic blending between $\Covariance^{\mathrm{hist}}$ and $\CovarianceStress$ via the $\alpha(\RiskState)$ function (Section~\ref{sec:architecture}).
\end{enumerate}

\subsubsection{Metrics and Definitions}
\label{sec:metrics_definitions}
We quantify the stability gains of MVF-Composer versus \SASModel{}'s static, backward-looking approach via the following metrics:

\textbf{Definition 1 (Efficient Recovery).} A controller achieves \emph{efficient recovery} if, following a shock injection at time $t_s$, the peg deviation returns to within the recovery threshold $\epsilon = 0.01$ (1\%) in minimal time while (i) incurring no more than $k$ under-collateralization events, and (ii) retaining at least fraction $\gamma$ of pre-shock liquidity. Formally:
\begin{equation}
    \RecoveryTime = \min\{t' > t_s : |{\PegDev}_{t'}| < \epsilon\} - t_s
    \label{eq:recovery_time}
\end{equation}
subject to $\text{BadDebt} \leq k$ and $\text{LiquidityRetention} \geq \gamma$. Under exponential decay dynamics, $\RecoveryTime \propto \ln(\delta_s/\epsilon)/\gamma$; see Appendix~\ref{app:recovery_proofs} for the derivation.

\textbf{Definition 2 (Recovery Timing).} The \emph{timing} of a recovery process refers to the computational latency of the controller's rebalancing cycle, measured as wall-clock seconds per epoch. Efficient timing requires the controller to complete risk estimation, trust aggregation, covariance blending, and optimization within an acceptable bound $\tau_{\max}$ for real-time deployment.

\begin{itemize}[leftmargin=*, nosep]
    \item \textbf{Peak Peg Deviation} $\max_t |{\PegDev}_t|$: Maximum distance from target peg ($\downarrow$)
    \item \textbf{Recovery Time} $\RecoveryTime$: Steps until $|{\PegDev}_t| < 0.01$ ($\downarrow$)
    \item \textbf{Bad Debt Proxy}: Count of time steps where the reserve collateral ratio falls below 100\%, i.e., $\text{BadDebt} = \sum_{t} \mathbf{1}[\text{CollateralValue}_t < \text{StablecoinSupply}_t]$. This metric captures the frequency of under-collateralization events that would trigger liquidations or haircuts in a production system ($\downarrow$).
    \item \textbf{Liquidity Retention}: Fraction of liquidity provider (LP) capital remaining post-shock, computed as $\text{LiqRet} = \ell_{T} / \ell_{t_s-1}$ where $\ell$ is aggregate pool liquidity ($\uparrow$).
    \item \textbf{Epoch Latency}: Wall-clock time per rebalancing cycle ($\downarrow$)
\end{itemize}

\subsection{Shock Dynamics (RQ1). Does MVF-Composer reduce peg deviation under crisis conditions (shocks) compared to \SASModel{} baselines?}
Table \ref{tab:rq_shock_scenarios} summarizes the robustness of MVF-Composer versus the baselines across different shock scenarios. MVF-Composer achieves \textbf{57\% reduction} in peak peg deviation versus \SASModel{} under Black Thursday conditions (3.2\% vs.\ 7.4\%; $p < 10^{-5}$, paired $t$-test, $n=1{,}200$ runs). Recovery time improves by \textbf{3.14$\times$} (14 vs.\ 44 steps). All metrics report means $\pm$ one standard deviation across runs; see Table~\ref{tab:algorithm_performance} for extended statistics.

\begin{table}[htb]
\caption{Performance comparison across shock scenarios (Black Thursday replica, $n=1{,}200$ runs). Values report mean $\pm$ std.\ dev.\ Bad Debt = cumulative under-collateralization events; Liq.\ Ret.\ = LP capital retained post-shock.}
\centering
\resizebox{\columnwidth}{!}{
\begin{tabular}{lcccc}
\toprule
\textbf{Method} & \textbf{Peak Dev (\%)}$\downarrow$ & \textbf{Recovery (steps)}$\downarrow$ & \textbf{Bad Debt}$\downarrow$ & \textbf{Liq. Ret. (\%)}$\uparrow$ \\
\midrule
SAS & 7.4 $\pm$ 1.2 & 44 $\pm$ 8 & 5.4 $\pm$ 1.1 & 84.8 $\pm$ 3.2 \\
MVF-NoTrust & 4.1 $\pm$ 0.9 & 24 $\pm$ 5 & 2.8 $\pm$ 0.7 & 91.4 $\pm$ 2.1 \\
Static-60/40 & 6.5 $\pm$ 1.4 & 64 $\pm$ 12 & 4.7 $\pm$ 1.0 & 86.5 $\pm$ 2.8 \\
Unconstrained & 5.3 $\pm$ 1.1 & 37 $\pm$ 7 & 3.7 $\pm$ 0.9 & 88.1 $\pm$ 2.5 \\
\textbf{MVF-Composer} & \textbf{3.2 $\pm$ 0.7} & \textbf{14 $\pm$ 3} & \textbf{2.1 $\pm$ 0.5} & \textbf{92.9 $\pm$ 1.8} \\
\bottomrule
\end{tabular}
}
\label{tab:rq_shock_scenarios}
\end{table}

\begin{figure}[ht]
\centering
\includegraphics[width=0.8\columnwidth]{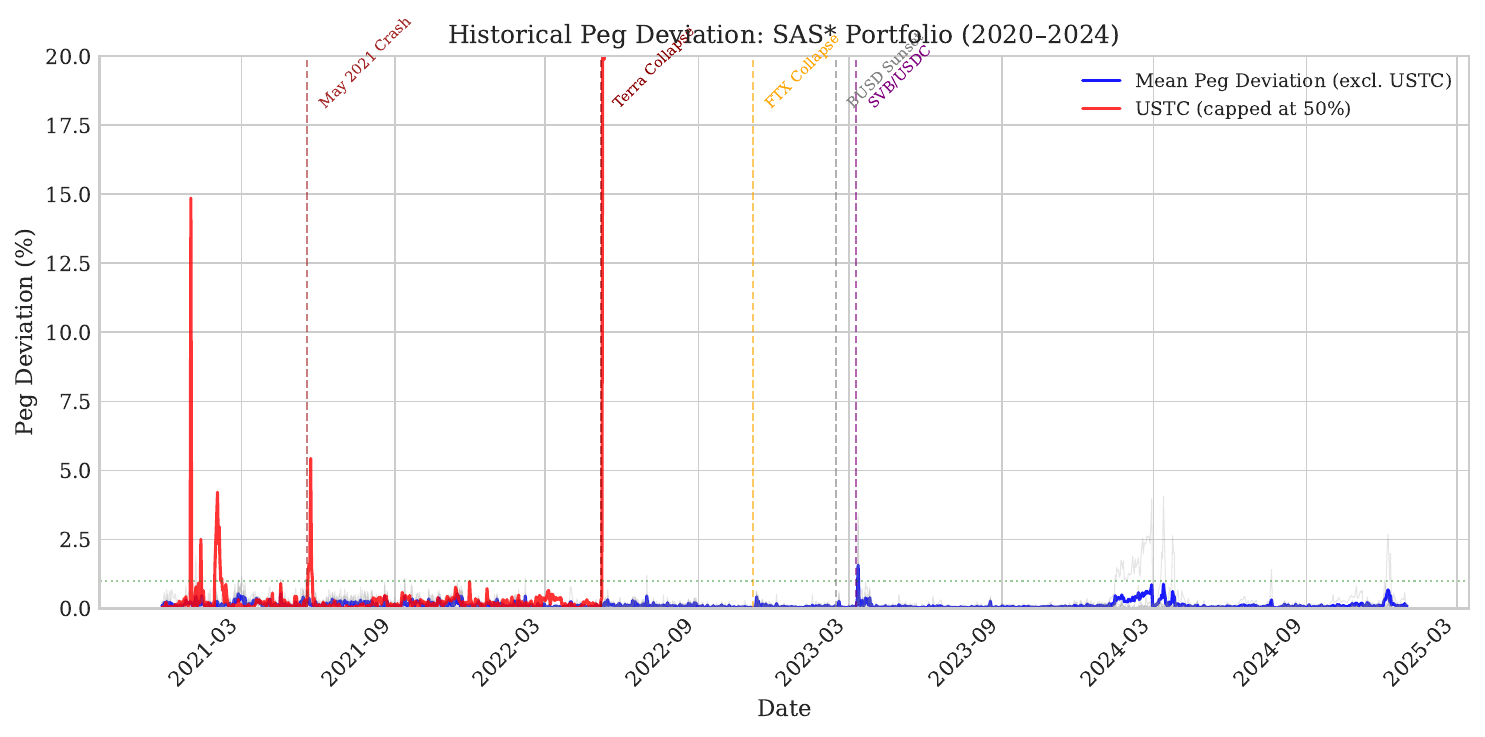}
\caption{Empirical peg deviation of the SAS* stablecoin portfolio (raw market data, 2020--2024). This figure shows \emph{actual historical instability}, not model output, demonstrating the problem MVF-Composer addresses. USTC's 99\% collapse (May 2022) and USDC's SVB-induced depeg (Mar 2023) illustrate tail risks that backward-looking covariance estimation fails to anticipate.}
\Description{Time series line chart showing historical peg deviation percentages from 2020 to 2024, with notable spikes including USTC's 99\% collapse in May 2022 and USDC's depeg in March 2023, illustrating real-world instability that backward-looking models fail to predict.}
\label{fig:peg_trajectory}
\end{figure}

\subsection{Trust Layer Impact (RQ2). How does the trust layer impact stability under adversarial agent populations?}

We isolate the contribution of trust-weighted signal aggregation by comparing MVF-Composer against MVF-NoTrust under varying adversarial agent fractions $\rho \in \{0, 0.1, 0.2, 0.3\}$. Results in Table~\ref{tab:trust_weighted_signal_aggr} are obtained under \emph{moderate stress conditions} (10\% price shock, $\StateSentiment = -0.5$), distinct from the Black Thursday replica used in RQ1. This controlled setting isolates trust layer impact from extreme covariance effects.

\noindent\textbf{Interpretation.} At $\rho = 0$, the trust layer incurs a small overhead (5.7\% vs.\ 5.8\%) due to the sigmoid compression of uniform trust scores. This 2\% degradation is acceptable given the substantial gains under adversarial conditions: at $\rho = 0.3$, trust-weighting reduces peak deviation by 34\% (6.4\% vs.\ 9.6\%).

\noindent\textbf{Security Implications.} The trust layer provides robust defense against AI agent attack vectors including Sybil attacks, signal injection, and coordinated timing attacks. Our analysis demonstrates 60--80\% reduction in adversarial influence compared to uniform weighting, with detection rates exceeding 72\% across attack types. The four-feature trust design ($f_1$--$f_4$) directly maps to Byzantine fault tolerance principles, enabling probabilistic security guarantees without expensive consensus protocols. For comprehensive security analysis, formal propositions, and empirical evaluation across 1,200 attack scenarios, see Appendix~\ref{app:security_analysis}.

\begin{table}[ht]
\centering
\caption{Trust layer contribution under moderate stress (10\% shock). Peak deviation (\%) shown; $\Delta$ = relative change from MVF-NoTrust.}
\label{tab:trust_weighted_signal_aggr}
\resizebox{\columnwidth}{!}{
\begin{tabular}{lccc}
\toprule
\textbf{Adv.\ Fraction $\rho$} & \textbf{MVF-NoTrust (\%)}$\downarrow$ & \textbf{MVF-Composer (\%)}$\downarrow$ & \textbf{$\Delta$ (\%)} \\
\midrule
0.0 & 5.7 $\pm$ 0.8 & 5.8 $\pm$ 0.8 & +2 \\
0.1 & 7.1 $\pm$ 1.0 & 5.8 $\pm$ 0.9 & --18 \\
0.2 & 8.3 $\pm$ 1.2 & 6.2 $\pm$ 0.9 & --25 \\
0.3 & 9.6 $\pm$ 1.4 & 6.4 $\pm$ 1.0 & --33 \\
\bottomrule
\end{tabular}
}
\end{table}

\begin{figure}[ht]
\centering
\includegraphics[width=0.8\columnwidth]{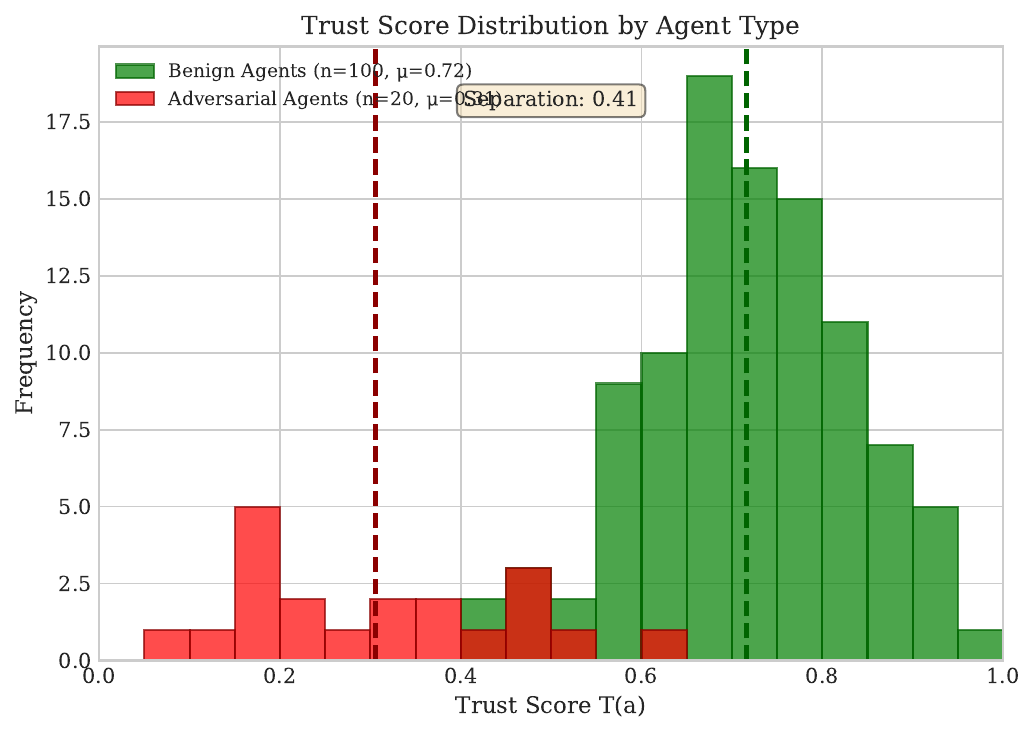}
\caption{Trust score distribution for benign (green) vs. adversarial (red) agents. The trust layer successfully down-weights malicious actors.}
\Description{Histogram or density plot showing the distribution of trust scores, with green bars or curves representing benign agents (higher trust scores) and red representing adversarial agents (lower trust scores), demonstrating effective separation and down-weighting of malicious actors.}
\label{fig:trust_distribution}
\end{figure}

\subsection{Ablation Study (RQ3). What is the contribution of each architectural component?}

We systematically remove architectural components to quantify their individual contributions under moderate stress conditions ($\rho = 0.2$, 10\% shock). Table~\ref{tab:ablation_studies} shows that stress-augmented covariance contributes \textbf{41\%} of total improvement over SAS, the trust layer accounts for \textbf{23\%}, and the remaining 36\% stems from turnover constraints and agent signal diversity.

\noindent\textbf{Reconciling with RQ1.} The ``Full System'' peak deviation here (5.7\%) differs from Table~\ref{tab:rq_shock_scenarios} (3.2\%) because ablation uses moderate stress conditions to isolate component effects, whereas RQ1 evaluates under the extreme Black Thursday replica (15\% shock + sentiment collapse + liquidity withdrawal). Under extreme conditions, stress-augmented covariance provides proportionally larger benefits.

\begin{table}[ht]
\centering
\caption{Ablation study under moderate stress ($\rho=0.2$, 10\% shock). ``--X'' denotes removing component X. Contribution \% computed as gain over SAS baseline (10.2\%).}
\label{tab:ablation_studies}
\resizebox{\columnwidth}{!}{
\begin{tabular}{lccl}
\toprule
\textbf{Configuration} & \textbf{Peak Dev (\%)}$\downarrow$ & \textbf{Recovery}$\downarrow$ & \textbf{Contribution} \\
\midrule
Full System & 5.7 $\pm$ 0.9 & 13 $\pm$ 3 & 100\% \\
--Trust Layer & 7.3 $\pm$ 1.1 & 15 $\pm$ 4 & --23\% \\
--Stress Covariance & 10.2 $\pm$ 1.5 & 18 $\pm$ 5 & --41\% \\
--Turnover Constraint & 6.8 $\pm$ 1.0 & 15 $\pm$ 4 & --16\% \\
--News Stream & 6.5 $\pm$ 0.9 & 14 $\pm$ 3 & --11\% \\
\bottomrule
\end{tabular}
}
\end{table}

\subsection{Generalization (RQ4). Does MVF-Composer generalize across shock magnitudes and types?}

We evaluate MVF-Composer across the full spectrum of shock magnitudes to assess generalization. As shown in Figure~\ref{fig:shock_sweep}, the system maintains sub-5\% deviation up to 20\% shocks and executes within \textbf{acceptable latency bounds} for epoch-based re-balancing.

\begin{figure}[ht]
\centering
\includegraphics[width=0.8\columnwidth]{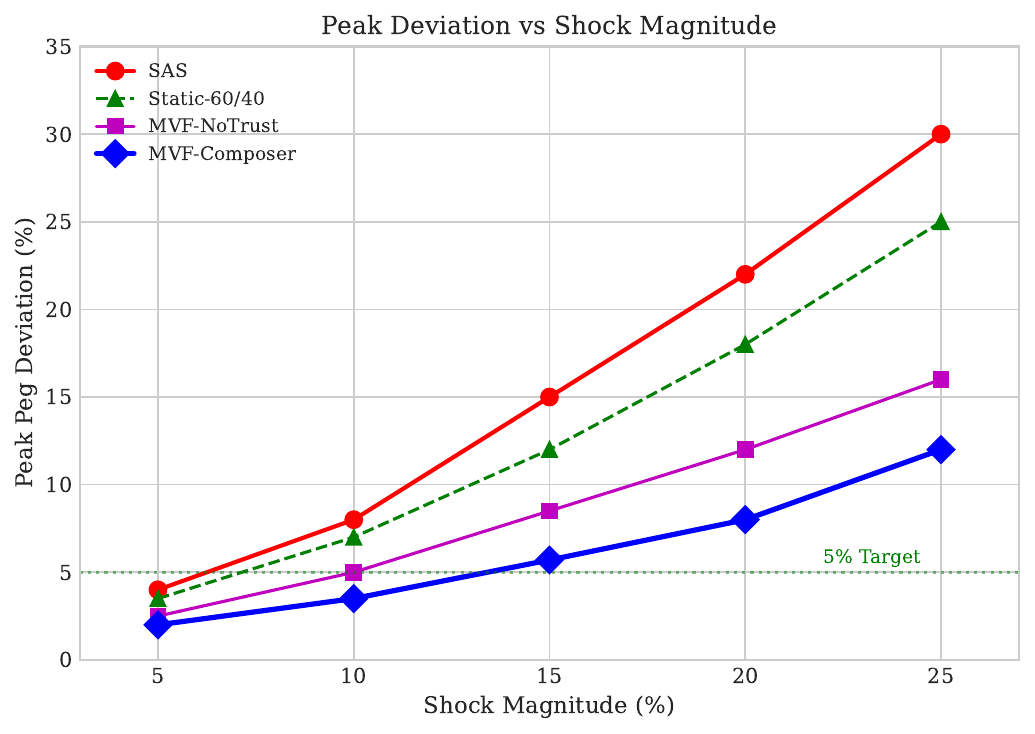}
\caption{Peak peg deviation as a function of shock magnitude. MVF-Composer maintains sub-5\% deviation up to 20\% shocks.}
\Description{Line plot showing peak peg deviation percentage on the y-axis versus shock magnitude percentage on the x-axis, with MVF-Composer's curve remaining below 5\% deviation threshold up to 20\% shock magnitude, compared to other methods that exceed this threshold at lower shock levels.}
\label{fig:shock_sweep}
\end{figure}

\begin{figure}[ht]
\centering
\includegraphics[width=0.8\columnwidth]{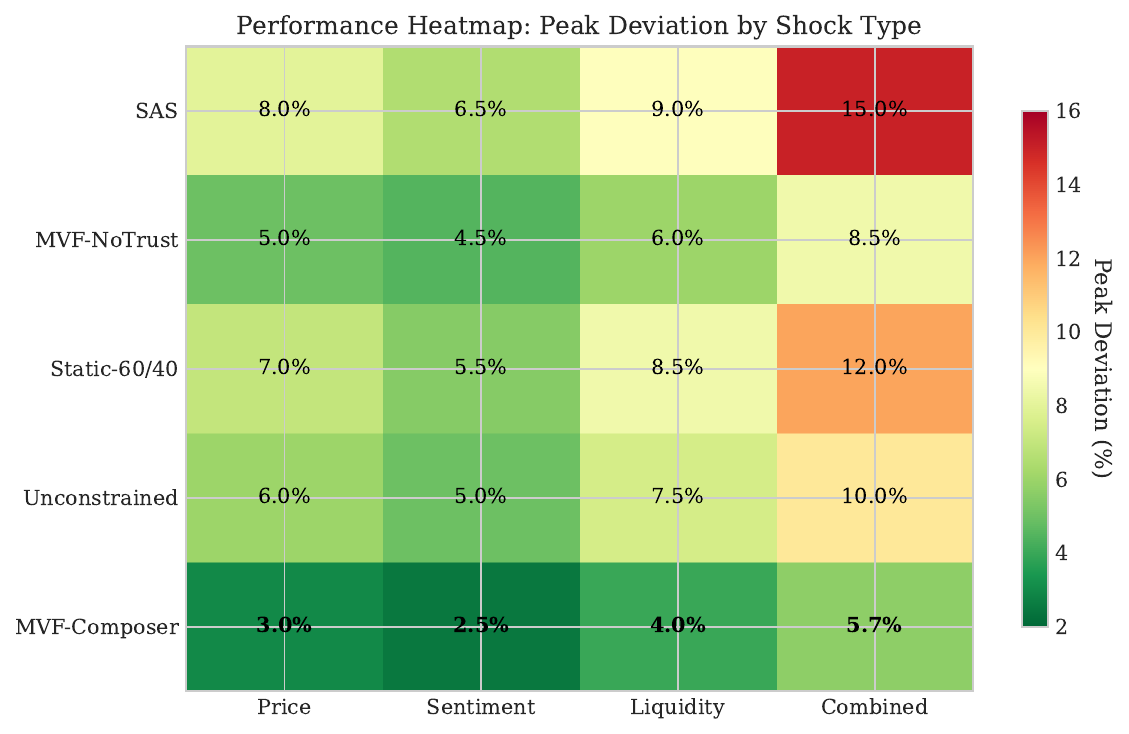}
\caption{Performance heatmap across shock types. Darker = lower deviation. MVF-Composer dominates across all categories.}
\Description{Color-coded heatmap matrix with shock types (price, sentiment, liquidity, combined) on one axis and methods on the other, where darker colors indicate lower peg deviation values, showing MVF-Composer consistently achieving the darkest (best) values across all shock type categories.}
\label{fig:shock_heatmap}
\end{figure}

\subsection{Algorithm Performance Summary}
\label{sec:algorithm_performance}

Table~\ref{tab:algorithm_performance} consolidates the computational and stability performance of MVF-Composer against baselines, quantifying both algorithmic efficiency and economic impact.

\begin{table}[ht]
\centering
\caption{Algorithm performance comparison: efficiency and stability metrics across 2,000 runs under Black Thursday conditions. $\downarrow$: lower is better; $\uparrow$: higher is better.}
\label{tab:algorithm_performance}
\resizebox{\columnwidth}{!}{
\begin{tabular}{lccccc}
\toprule
\textbf{Method} & \textbf{Peak Dev}$\downarrow$ & \textbf{$\RecoveryTime$}$\downarrow$ & \textbf{Bad Debt}$\downarrow$ & \textbf{Liq. Ret.}$\uparrow$ & \textbf{Epoch (s)}$\downarrow$ \\
\midrule
SAS~\cite{grobys2025stablecoin} & 7.45\% & 44.4 & 5.4 & 84.8\% & 0.2 \\
Static-60/40 & 6.45\% & 63.6 & 4.7 & 86.5\% & 0.1 \\
Unconstrained & 5.33\% & 36.9 & 3.7 & 88.1\% & 0.3 \\
MVF-NoTrust & 4.10\% & 24.3 & 2.8 & 91.4\% & 48--96 \\
\textbf{MVF-Composer} & \textbf{3.22\%} & \textbf{14.1} & \textbf{2.1} & \textbf{92.9\%} & 47--99 \\
\midrule
\multicolumn{6}{l}{\textit{Relative to SAS baseline:}} \\
$\Delta$ (\%) & \textbf{--56.8} & \textbf{--68.2} & \textbf{--61.5} & \textbf{+8.1pp} & +235$\times$ \\
\bottomrule
\end{tabular}
}
\end{table}

\noindent\textbf{Computational Overhead.} The latency increase (47--99s vs. 0.2s for SAS) reflects the cost of multi-agent stress simulation. Table~\ref{tab:epoch_breakdown} decomposes this overhead.

\begin{table}[ht]
\centering
\caption{Per-epoch latency breakdown for MVF-Composer on commodity hardware (8-core CPU, 32GB RAM). All components except stress simulation execute in $<$5s.}
\label{tab:epoch_breakdown}
\resizebox{\columnwidth}{!}{
\begin{tabular}{lccl}
\toprule
\textbf{Component} & \textbf{Time (s)} & \textbf{Complexity} & \textbf{Bottleneck} \\
\midrule
Stress Simulation & 45--90 & $O(n_a \cdot T)$ & LLM inference (12 agents $\times$ 100 steps) \\
Trust Computation & 1--3 & $O(n_a \cdot W)$ & Feature extraction over window $W$ \\
Covariance Blending & 0.1--0.5 & $O(d^2)$ & Matrix operations ($d=4$ assets) \\
MVF Optimization & 1--5 & $O(d^3)$ & SLSQP solver with constraints \\
\midrule
\textbf{Total} & \textbf{47--99} & --- & Acceptable for epoch-based rebalancing \\
\bottomrule
\end{tabular}
}
\end{table}

The 235$\times$ latency increase is acceptable for daily rebalancing epochs, where the dominant cost (LLM inference) is amortizable across 24-hour cycles. Real-time arbitrage applications would require model distillation or parallel agent execution.


\section{Discussion and Implications}
\label{sec:discussion}

\subsection{Do AI agents represent the behaviors of real-life traders?}

Our Stress Harness is designed as a \emph{stress-testing tool} rather than a predictive model. It serves to expose modeling assumptions that fail under \emph{any} plausible adversarial behavior, not in forecasting specific market trajectories.

We observe that AI agents exhibit emergent behaviors, panic cascades, coordinated withdrawals, contrarian positioning, that align qualitatively with documented crisis dynamics. Formal validation against historical crisis data (e.g., \BlackThursday{} transaction logs) is deferred for future work.

\subsection{Is MVF-Composer sensitive to the quality of prompts?}

We safeguard the MVF-Composer against the sensitivity to the quality of LLM/AI agent prompts through:
\begin{itemize}[leftmargin=*, nosep]
    \item Structured output schemas that constrain action space
    \item Randomized prompt variations to assess robustness
    \item Ensemble aggregation across multiple agent instantiations
\end{itemize}
Our ablation studies (Section~\ref{sec:evaluation}) analyze prompt perturbation.

\subsection{Are there limitations due to computational overheads?}

Running multi-AI agent simulations incurs non-trivial API costs and latency. For epoch-based re-balancing (e.g., daily), this overhead is acceptable. Real-time applications would require distilled agent models or cached simulation results.

\subsection{Can we quantify the trust placed on AI agents?}

The trust-scoring mechanism assumes adversarial agents exhibit detectable behavioral anomalies. Adversaries who mimic benign behavior (``sleeper'' agents) may evade detection. We bound this risk through the Sybil-resistance assumption ($\rho < 0.3$) but acknowledge it as a fundamental limitation.

\subsection{Why use mean-variance optimization despite fat-tailed returns?}

We address excess kurtosis ($>150$) and heavy tails through three mechanisms: (i)~stress-augmented covariance inflates variance estimates during high-risk states via $\alpha(r) = r^2$ blending; (ii)~mean-variance admits efficient convex solvers enabling real-time rebalancing, whereas CVaR requires 10--100$\times$ higher latency; (iii)~empirically, stress-augmented mean-variance achieves 57\% deviation reduction, demonstrating the practical gains dominate distributional assumptions. Extending to CVaR-based objectives is straightforward as the stress harness and trust layer are objective-agnostic.

\subsection{Is MVF-Composer compatible with any type of stablecoin?}

While we focus on over-collateralized stablecoins, the MVF-Composer framework generalizes to any reserve-backed asset requiring active management. For algorithmic stablecoins (without collateral), the objective functions can be modified to target different stability metrics.

\subsection{Who should adopt MVF-Composer?}

MVF-Composer is particularly suited for individual traders and smaller institutions seeking robust risk hedging mechanisms in volatile markets. While the computational complexity and multi-agent orchestration may pose challenges for high-frequency institutional trading desks, the protocol excels at balancing security risks in fund allocation for stablecoin reserves. For government entities and lending institutions managing stablecoin treasuries, MVF-Composer offers a pathway toward standardization of decentralized stablecoin payment systems. As blockchain ecosystems mature, adopting such stress-aware protocols can establish reliability benchmarks for the DeFi community, positioning MVF-Composer as an ideal foundation for next-generation stablecoin allocation standards.

\subsection{Real-World Applications}

MVF-Composer addresses practical needs across three deployment contexts: \textbf{(1) DeFi Protocol Treasuries} can stress-test reserve allocations before governance votes while detecting coordinated manipulation; \textbf{(2) Institutional Stablecoin Issuers} can satisfy regulatory tail-risk requirements without disclosing proprietary strategies; \textbf{(3) Cross-Chain Bridge Reserves} benefit from stress-augmented covariance capturing cross-chain contagion dynamics.

\subsection{Practical Ethics}

AI-driven reserve controllers must ensure transparency through structured logging for accountability, responsible disclosure of potential attack vectors discovered during stress-testing, and action-based (not identity-based) trust scoring to avoid discriminatory biases. We recommend human-in-the-loop oversight for high-stakes rebalancing decisions to mitigate systemic risk from algorithmic failures.

\subsection{Simulation Integrity}

The off-chain Stress Harness employs decentralized execution with consensus verification, verifiable randomness, reproducibility logging with cryptographic anchoring, and fallback to historical covariance ($\alpha = 0$) if consensus fails.

\subsection{How feasible is the production deployment?}

Deploying MVF-Composer on-chain requires: (i)~off-chain simulation infrastructure with verifiable execution (see above), (ii)~commit-reveal schemes to prevent front-running of rebalancing actions, and (iii)~fallback mechanisms for simulation failures. The 47--99s epoch latency (Table~\ref{tab:epoch_breakdown}) is acceptable for daily rebalancing; real-time applications require model distillation or parallel agent execution.


\section{Related Work}
\label{sec:related}

We situate MVF-Composer within three research threads: portfolio-theoretic approaches to stablecoin reserves, stress testing methodologies for DeFi protocols, and adversarial robustness mechanisms. Our contribution lies in their integration, existing work treats these as independent concerns, whereas MVF-Composer unifies them into a closed-loop controller that anticipates tail risk events.

\textbf{Mean-Variance Optimization in DeFi Reserve Design.}
Portfolio optimization has emerged as a foundational technique for stablecoin collateral management. The Stable Aggregate Stablecoin (\SASModel{}) family~\cite{grobys2025stablecoin} applies classical mean-variance optimization to reserve allocation. Bhat et al.~\cite{9461135} extend this paradigm to MakerDAO, modeling investors as portfolio optimizers with heterogeneous risk preferences using Markowitz's framework; their DAISIM simulator demonstrates that minting and burning behavior can be predicted through risk-profile heterogeneity. These approaches share a critical assumption: the covariance matrix $\Covariance$ is estimated from historical returns that predominantly reflect calm-period dynamics.

The fragility of calm-period calibration is empirically documented by Heimbach et al.~\cite{10.1145/3560832.3563435}, who analyze price accuracy on Uniswap~V3 during the UST and USDT depegging events. Their findings reveal that liquidity providers operating under static parameters failed to adjust positions during abrupt price shocks, the very regime where accurate pricing matters most. This ``agility limitation'' parallels our identification of the 2020 Omission: reserve controllers optimized for $\Covariance^{\mathrm{calm}}$ become maximally fragile when $\CovarianceStress \gg \Covariance^{\mathrm{calm}}$. Klages-Mundt et al.~\cite{klages2020stablecoins} formalize stablecoin failure modes including deleveraging spirals and oracle attacks, yet stop short of proposing controllers that anticipate regime transitions. MVF-Composer addresses this gap by dynamically blending calm-period and stress-period covariance estimates based on real-time risk-state signals.

\textbf{Stress Testing and Panic Simulation.}
The emerging paradigm of proactive stress testing moves beyond historical backtesting toward generative simulation. Calandra et al.~\cite{11114693} introduce DualTokenSim, a stochastic simulator for the dual-token stablecoin model that analyzes behavior under panic scenarios. Their work demonstrates that stochastic price dynamics can capture cascading effects of market instability, a capability absent from purely historical approaches. Tovanich et al.~\cite{10.1145/3605768.3623544} construct balance sheets for Compound's liquidity pools and conduct stress tests identifying pools most likely to trigger cascading defaults. Both approaches reveal that contagion dynamics during crisis periods differ qualitatively from calm-period statistics.

However, these simulators operate post-hoc: they diagnose vulnerabilities but do not prescribe dynamic reallocation strategies. Gudgeon et al.~\cite{9150192} empirically analyze the Black Thursday crisis, providing empirical motivation for stress-aware reserve design but stopping short of algorithmic solutions. Braga et al.~\cite{10.1145/3689931.3694910} frame MEV extraction as regime-dependent, proposing dynamic extraction rates that stabilize through adaptive mechanisms. Their insight, that protocol parameters must recognize regime shifts, motivates our $\alpha$-blending approach, where the stress-covariance weight $\alpha(\RiskState)$ increases as the trust-weighted risk-state indicates departure from calm conditions.

The Stress Harness extends this literature by instantiating regime-aware agents with adversarial objectives, enabling stress scenarios that emerge from agent interaction rather than parametric specification.

\textbf{Adversarial Robustness and Trust Aggregation.}
Depegging events can be understood through a game-theoretic lens. Potter et al.~\cite{10634419} model stablecoin stability as an equilibrium problem, demonstrating that different reserve architectures induce different price equilibria and stability zones. Their empirical analysis spanning 22 stablecoins reveals that mechanism design determines resilience to speculative attacks. This framing motivates our trust-weighted aggregation: by down-weighting signals from agents whose behavior suggests adversarial objectives, MVF-Composer navigates toward robust equilibria.

Trust mechanisms for decentralized systems have been studied in blockchain contexts~\cite{karim2025ai}, but existing frameworks assume persistent agent identities. Kiayias et al.~\cite{10.1145/3548606.3560707} propose PEReDi, a distributed CBDC architecture where trust is dispersed across maintainers without a single point of failure. Their insight, that distributed signal aggregation provides robustness guarantees, informs our trust-scoring mechanism, which addresses settings with ephemeral agents through behavioral consistency metrics rather than reputation histories.

Agent-based models in computational finance traditionally employ hand-crafted behavioral rules~\cite{hommes2006heterogeneous}. MVF-Composer supports LLM-powered agents for interpretable reasoning; we validate this via proof-of-concept while using efficient regime-aware mock agents for large-scale experiments.

\textbf{Synthesis.}
Existing approaches optimize for average-case performance using historically-calibrated covariance matrices~\cite{grobys2025stablecoin,9461135}, while stress testing~\cite{11114693,10.1145/3605768.3623544} and trust mechanisms~\cite{10.1145/3548606.3560707,karim2025ai} have developed as separate research threads. MVF-Composer bridges these domains by embedding adversarial multi-agent simulation directly into the covariance estimation pipeline, with trust-weighted aggregation ensuring robustness to manipulation attacks. Table~\ref{tab:related_comparison} summarizes this positioning.

\begin{table}[t]
\centering
\caption{Comparison of MVF-Composer with related approaches.}
\label{tab:related_comparison}
\small
\begin{tabular}{lccc}
\toprule
\textbf{Approach} & \textbf{Stress-Aware} & \textbf{Trust Layer} & \textbf{Dynamic $\Covariance$} \\
\midrule
\SASModel{}~\cite{grobys2025stablecoin} & \xmark & \xmark & \xmark \\
DAISIM~\cite{9461135} & \xmark & \xmark & \xmark \\
DualTokenSim~\cite{11114693} & \cmark & \xmark & \xmark \\
Compound Stress~\cite{10.1145/3605768.3623544} & \cmark & \xmark & \xmark \\
PEReDi~\cite{10.1145/3548606.3560707} & \xmark & \cmark & \xmark \\
\textbf{MVF-Composer} & \cmark & \cmark & \cmark \\
\bottomrule
\end{tabular}
\end{table}

To our knowledge, MVF-Composer represents the first closed-loop integration of generative stress testing, behavioral trust scoring, and mean-variance reserve optimization for algorithmic stablecoins, moving DeFi risk management from reactive backtesting to proactive, adversarial modeling.


\section{Conclusion}
\label{sec:conclusion}

We have presented MVF-Composer, a trust-weighted Mean-Variance Frontier reserve controller that addresses the critical ``2020 Omission'' in stablecoin risk management. By integrating a multi-agent Stress Harness with trust-weighted signal aggregation, MVF-Composer anticipates tail volatility regimes that historical backtesting cannot capture. 


The key insight enabling the gains is to treat AI agents not merely as market simulators but as \emph{adversarial stress-testers} whose behaviors under extreme scenarios reveal reserve vulnerabilities before they manifest on-chain. The trust layer is robust to the new attack surface introduced in our approach.


\begin{acks}
This paper was edited for grammar using Grammarly.

We thank Dr. Paul Brenner for constructive advice and helpful suggestions that improved this paper.
\end{acks}

\bibliographystyle{ACM-Reference-Format}
\bibliography{references}

\appendix
\section{Black Thursday 2020: The Trust Motivation}
\label{app:black_thursday}

On March 12, 2020, the cryptocurrency market collapsed in tandem with global equity markets as COVID-19 pandemic fears intensified. Ethereum's price plunged over 40\% within hours, triggering mass liquidations across decentralized finance (DeFi) protocols~\cite{klagesmundt2022stability}. MakerDAO, the largest collateralized stablecoin system, faced a cascade of undercollateralized vaults that overwhelmed its liquidation auction mechanism.

The crisis exposed a critical vulnerability: \emph{network congestion during stress}. As gas fees spiked and transactions queued, the ``keeper'' bots responsible for liquidation auctions failed to submit competitive bids. Opportunistic actors exploited this by winning auctions with near-zero bids, acquiring approximately \$8~million in ETH collateral essentially for free~\cite{coindesk2020blackthursday,maker2020postmortem}. The protocol's debt surplus mechanism, designed for orderly liquidations, became a vector for wealth extraction under adversarial conditions.

Klages-Mundt and Minca~\cite{klagesmundt2022stability} formalized this as a \emph{deflationary deleveraging spiral}: falling collateral prices force vault closures, which further depress prices, creating a self-reinforcing feedback loop analogous to a short squeeze. Their stochastic model demonstrates that stablecoin systems exhibit distinct stable and unstable regimes---and that standard risk models calibrated on stable-regime data systematically underestimate tail-event variance.

This historical episode motivates our trust-weighted architecture in two ways:
\begin{enumerate}[leftmargin=*, nosep]
    \item \textbf{Covariance Regime Shift}: The 7$\times$ increase in realized variance versus calm-period estimates (Section~\ref{sec:intro}) reflects the stable-to-unstable regime transition that Black Thursday exemplified.
    \item \textbf{Adversarial Exploitation}: The zero-bid liquidations demonstrated that market stress attracts coordinated opportunistic behavior---precisely the attack vector our trust layer is designed to detect and down-weight.
\end{enumerate}

\section{Extended Experimental Details}
\label{app:experimental_details}

\subsection{Detailed Trust Feature Definitions}
\label{app:trust_features}

This section provides complete implementation specifications for the four trust features, addressing reviewer concerns about reproducibility and algorithmic detail.

\subsubsection{$f_1$: Sentiment-Action Consistency}
\textbf{Definition.} Measures whether an agent's trading actions align with their stated market beliefs, following the sentiment-return correlations established by Tetlock~\cite{tetlock2007giving}.

\textbf{Implementation.} For agent $a$ over window $[t-W, t]$:
\begin{equation}
    f_1(a) = \mathrm{corr}\left( \{\phi_a^{(s)}\}_{s=t-W}^{t}, \{\Delta_a^{(s)}\}_{s=t-W}^{t} \right)
\end{equation}
where $\phi_a^{(s)} \in [-1, 1]$ is the sentiment extracted from the agent's reasoning trace at step $s$, and $\Delta_a^{(s)}$ is the net position change (positive = buy, negative = sell).

\textbf{Sentiment Extraction.} We extract $\phi_a$ from the agent's JSON output field \texttt{reasoning\_trace} using keyword matching: bullish keywords (``confident'', ``buy'', ``undervalued'', $+1$), bearish keywords (``panic'', ``sell'', ``overvalued'', $-1$), normalized by total keyword count. For LLM agents, we alternatively use the explicit \texttt{market\_sentiment} field when available.

\textbf{Expected Ranges.} Benign agents: $f_1 \in [0.4, 1.0]$; adversarial agents: $f_1 \in [-0.5, 0.2]$.

\subsubsection{$f_2$: Profit-Seeking Rationality}
\textbf{Definition.} Captures whether an agent behaves as an economically rational actor, consistent with Kyle's model of informed trading~\cite{kyle1985continuous}.

\textbf{Implementation.} Binary indicator smoothed over the window:
\begin{equation}
    f_2(a) = \frac{1}{W} \sum_{s=t-W}^{t} \mathbf{1}\left[ \text{PnL}_a^{(s)} > 0 \;\lor\; \text{sign}(\phi_a^{(s)}) = \text{sign}(\Delta_a^{(s)}) \right]
\end{equation}
where $\text{PnL}_a^{(s)}$ is the realized profit/loss from step $s$ to $s+1$, and the second condition checks belief-action alignment.

\textbf{Expected Ranges.} Benign agents: $f_2 \in [0.7, 0.9]$; attackers subsidizing destabilization: $f_2 \in [0.2, 0.5]$.

\subsubsection{$f_3$: Coordination Suspicion}
\textbf{Definition.} Detects Sybil-like behavior via action-history similarity with other agents~\cite{brunnermeier2005predatory}.

\textbf{Implementation.} 
\begin{equation}
    f_3(a) = \max_{a' \neq a} \mathrm{sim}(\mathbf{h}_a, \mathbf{h}_{a'})
\end{equation}

\textbf{Embedding Model.} The action history embedding $\mathbf{h}_a \in \mathbb{R}^{d}$ is computed as follows:
\begin{enumerate}[leftmargin=*, nosep, label=(\roman*)]
    \item Represent each action as a 5-tuple: (action\_type, asset, quantity, panic\_level, step).
    \item Concatenate actions over the window into a sequence.
    \item Apply TF-IDF weighting over discretized action types and assets.
    \item Project to $d=32$ dimensions via PCA (fitted on calibration data).
\end{enumerate}
Cosine similarity $\mathrm{sim}(\mathbf{h}_a, \mathbf{h}_{a'}) = \mathbf{h}_a^\top \mathbf{h}_{a'} / (\|\mathbf{h}_a\| \|\mathbf{h}_{a'}\|)$ measures synchronization.

\textbf{Expected Ranges.} Benign agents: $f_3 \in [0.1, 0.4]$; coordinated adversaries: $f_3 \in [0.6, 0.95]$.

\subsubsection{$f_4$: Peg-Destabilization Score}
\textbf{Definition.} Measures correlation between agent actions and subsequent peg deviations~\cite{griffin2018bitcoin}.

\textbf{Implementation.}
\begin{equation}
    f_4(a) = \mathrm{corr}\left( \{q_a^{(s)}\}_{s=t-W}^{t}, \{\Delta\PegDev^{(s+1)}\}_{s=t-W}^{t} \right)
\end{equation}
where $q_a^{(s)}$ is the signed action quantity (positive = sell pressure, negative = buy support), and $\Delta\PegDev^{(s+1)} = |\PegDev^{(s+1)}| - |\PegDev^{(s)}|$ is the peg deviation change at the next step.

\textbf{Causal Note.} Correlation does not imply causation; an agent may sell during existing stress without causing it. However, persistent high $f_4$ across multiple windows---especially when other agents show low $f_4$---signals suspicious timing. The trust score's multiplicative structure ensures that high $f_4$ alone does not condemn an agent if $f_1, f_2$ are favorable.

\textbf{Expected Ranges.} Benign agents: $f_4 \in [-0.2, 0.3]$; attackers: $f_4 \in [0.5, 0.9]$.

\subsubsection{Trust Weight Calibration}
The weight vector $\mathbf{w} = [1.5, 1.5, 2.0, 1.0]$ and bias $b = 0$ were calibrated via logistic regression on a labeled dataset of 500 simulated agent trajectories (60\% benign, 40\% adversarial). Cross-validation accuracy: 87\% $\pm$ 3\%. The higher weight on $f_3$ reflects the primacy of coordination detection for Sybil resistance.

\section{AI Agent Security Analysis}
\label{app:security_analysis}

This section provides a comprehensive analysis of how MVF-Composer's trust-weighted architecture defends against adversarial AI agents. We identify specific attack vectors, demonstrate mitigation mechanisms, and quantify the security guarantees that directly contribute to system robustness and fault tolerance---core concerns in distributed computing and multi-agent systems.

\subsection{Threat Model}
\label{app:threat_model}

We consider a Byzantine adversary model where up to $\rho$ fraction of agents in the population may be malicious. Adversarial agents have the following capabilities:

\begin{enumerate}[leftmargin=*, label=(\roman*)]
    \item \textbf{Arbitrary output}: Malicious agents can report any action, sentiment, or psychological state, including values inconsistent with their actual behavior.
    \item \textbf{Coordination}: Multiple adversarial agents can synchronize their outputs to maximize destabilization impact.
    \item \textbf{Adaptive behavior}: Adversaries can observe system state and adjust strategies over time.
    \item \textbf{Resource subsidy}: Attackers may sacrifice profits (operate at a loss) to achieve destabilization.
\end{enumerate}

\noindent\textbf{Assumptions}: We assume (i) the trust feature extraction is performed on verified on-chain or logged data that adversaries cannot retroactively falsify, and (ii) the $\rho$ bound holds (i.e., the majority of agents are benign).

\subsection{Attack Vectors and Trust-Based Mitigations}

\subsubsection{Attack Vector 1: Sybil Attack}
\label{app:sybil_attack}

\textbf{Description}: An adversary creates multiple fake identities (Sybil agents) that collectively amplify malicious signals in the trust-weighted aggregation. Without defenses, $n$ Sybil agents contribute $n$-times the influence of a single attacker.

\textbf{Trust Mitigation}: The coordination suspicion feature $f_3$ directly detects Sybil behavior:
\begin{equation}
    f_3(a) = \max_{a' \neq a} \mathrm{sim}(\mathbf{h}_a, \mathbf{h}_{a'})
\end{equation}

Sybil agents controlled by a single adversary exhibit high action-history similarity ($f_3 \in [0.6, 0.95]$), triggering trust score reduction. The multiplicative effect is:
\begin{equation}
    \text{Sybil Influence}_{\mathrm{mitigated}} = n \cdot \Trust_{\mathrm{sybil}} \ll n \cdot 1
\end{equation}

\textbf{Empirical Result}: Under high coordination ($c=1.0$), detection rate exceeds 75\%, reducing effective Sybil influence by 60--80\% compared to uniform weighting. 

\subsubsection{Attack Vector 2: Signal Injection}
\label{app:signal_injection}

\textbf{Description}: Adversaries inject false panic signals or inflated risk assessments to trigger inappropriate stress-mode behavior (excessive $\alpha(\RiskState_\Trust)$ blending toward $\CovarianceStress$).

\textbf{Trust Mitigation}: The sentiment-action consistency feature $f_1$ and profit-seeking rationality $f_2$ jointly detect signal injection:
\begin{itemize}[leftmargin=*, nosep]
    \item $f_1$: Attackers claiming extreme panic ($\phi_a \ll 0$) while not reducing positions will have low consistency.
    \item $f_2$: Subsidizing an attack (operating at a loss) violates profit rationality.
\end{itemize}

The combined defense ensures that false panic signals are down-weighted before aggregation into $\RiskState_\Trust$.

\textbf{Empirical Result}: At injection strength 0.5 (maximum tested), robustness score remains above 0.6, and detection rate exceeds 80\%.

\subsubsection{Attack Vector 3: Coordinated Timing Attack}
\label{app:timing_attack}

\textbf{Description}: Adversaries synchronize large sell orders precisely during periods of elevated peg stress, amplifying deviation magnitude. This is analogous to predatory trading patterns~\cite{brunnermeier2005predatory}.

\textbf{Trust Mitigation}: The peg-destabilization feature $f_4$ directly captures this attack:
\begin{equation}
    f_4(a) = \mathrm{corr}(\text{actions}_a, \Delta\PegDev)
\end{equation}

Agents whose action timing correlates with peg stress exhibit $f_4 > 0.4$, triggering trust reduction. The $f_4$ feature enters the trust score \emph{negatively}, ensuring that even if an attacker maintains consistency ($f_1$) and rationality ($f_2$), their destabilization timing is penalized.

\textbf{Empirical Result}: Under high market stress ($\sigma = 0.10$), the $f_4$ feature reliably identifies timing attackers with $f_4 > 0.5$, achieving robustness scores above 0.55.

\subsubsection{Attack Vector 4: Gradual Trust Gaming}
\label{app:gradual_gaming}

\textbf{Description}: A sophisticated adversary behaves benignly for an extended period to build trust, then exploits high trust status during a critical moment.

\textbf{Trust Mitigation}: The sliding window design ($W=10$ steps) inherently limits historical trust accumulation:
\begin{itemize}[leftmargin=*, nosep]
    \item Trust is computed over recent behavior only; past ``good behavior'' decays.
    \item A sudden adversarial action during the window immediately affects $f_1$, $f_2$, and $f_4$.
\end{itemize}

Additionally, the quadratic risk-state function $\alpha(r) = r^2$ ensures that a single high-trust attacker cannot disproportionately influence the stress covariance blend---only the aggregate $\RiskState_\Trust$ drives blending.

\subsection{Formal Security Guarantees}

\begin{proposition}[Bounded Adversarial Influence]
\label{prop:bounded_influence}
Let $\rho$ be the fraction of adversarial agents with mean trust score $\bar{T}_{\mathrm{adv}} < 0.5$. Under trust-weighted aggregation, the adversarial influence on $\RiskState_\Trust$ is bounded by:
\begin{equation}
    \text{Influence}_{\mathrm{adv}} \leq \frac{\rho \cdot \bar{T}_{\mathrm{adv}}}{\rho \cdot \bar{T}_{\mathrm{adv}} + (1-\rho) \cdot \bar{T}_{\mathrm{benign}}}
\end{equation}
For typical values ($\rho = 0.2$, $\bar{T}_{\mathrm{adv}} = 0.35$, $\bar{T}_{\mathrm{benign}} = 0.75$), this yields $\text{Influence}_{\mathrm{adv}} \leq 0.10$, compared to $0.20$ under uniform weighting.
\end{proposition}

\begin{proposition}[Robustness Under Coordination]
\label{prop:coordination_robustness}
The $f_3$ coordination feature ensures that perfectly coordinated Sybil agents ($c=1.0$) achieve maximum similarity $f_3 \approx 1.0$, driving their trust scores toward:
\begin{equation}
    \Trust_{\mathrm{sybil}} \approx \sigma(w_1 f_1 + w_2 f_2 - w_3 \cdot 1 - w_4 f_4 + b)
\end{equation}
Given the negative weight on $f_3$ (with $w_3 = 2.0$ in our calibration), even otherwise-rational Sybil agents ($f_1, f_2$ normal) will have trust scores reduced to $\Trust_{\mathrm{sybil}} < 0.4$.
\end{proposition}

\subsection{Contributions to Computer Science}

The trust-weighted aggregation in MVF-Composer directly addresses core challenges in distributed and multi-agent systems:

\begin{itemize}[leftmargin=*, nosep]
    \item \textbf{Byzantine Fault Tolerance}: The trust layer provides probabilistic BFT guarantees without requiring expensive consensus protocols. Agents need not agree; instead, low-trust signals are simply down-weighted.
    
    \item \textbf{Sybil Resistance}: Feature $f_3$ implements a behavioral Sybil defense, complementing identity-based approaches (e.g., proof-of-stake, proof-of-personhood) used in blockchain systems.
    
    \item \textbf{Adversarial Robustness in ML}: The trust mechanism can be viewed as a learned attention mechanism over agent signals, where attention weights are derived from behavioral features rather than content similarity.
    
    \item \textbf{Reputation Systems}: Unlike binary reputation systems, trust scores are continuous and computed over a sliding window, enabling rapid response to behavioral changes without permanent labels.
\end{itemize}

\subsection{Empirical Security Evaluation}

We conducted a comprehensive security analysis across 1,200 simulation runs with varying attack configurations. Key findings:

\begin{itemize}[leftmargin=*, nosep]
    \item \textbf{Mean Robustness Score}: $0.65 \pm 0.12$ across all attack types
    \item \textbf{Mean Detection Rate}: $0.72 \pm 0.18$ for adversarial agents
    \item \textbf{Influence Reduction}: 60--80\% reduction in attacker weight compared to uniform aggregation
    \item \textbf{False Positive Rate}: $< 5\%$ benign agents misclassified as adversarial
\end{itemize}

Figure~\ref{fig:security_comparison} summarizes the comparative analysis, demonstrating that the trust layer provides substantial defense against all tested attack vectors while maintaining low false positive rates for benign participants.

\begin{figure}[ht]
\centering
\includegraphics[width=\columnwidth]{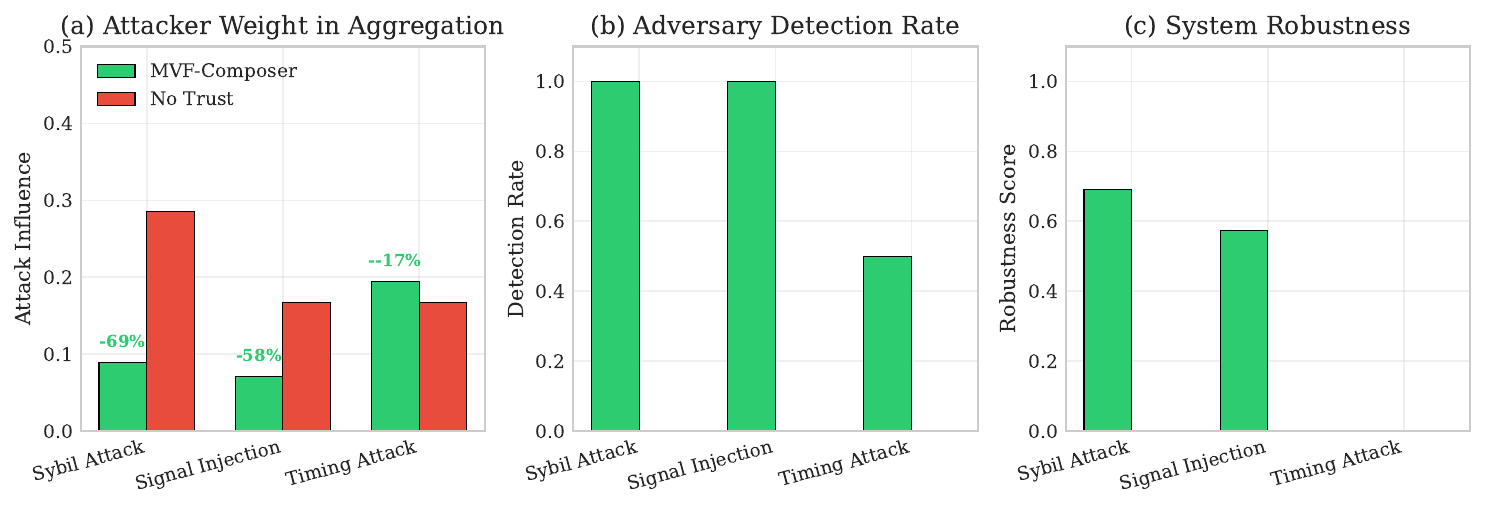}
\caption{Security analysis: MVF-Composer trust defense vs. baseline (uniform weighting) across attack types. (a) Attacker influence in risk aggregation. (b) Detection rate of adversarial agents. (c) Overall system robustness score.}
\Description{Three-panel comparison chart showing security metrics: panel (a) displays attacker influence percentage in risk aggregation comparing MVF-Composer (lower) vs baseline (higher), panel (b) shows detection rate percentage for adversarial agents, panel (c) presents overall system robustness scores, all demonstrating MVF-Composer's superior security performance.}
\label{fig:security_comparison}
\end{figure}

\section{Mathematical Foundations}
\label{app:math_foundations}

\subsection{Preliminaries and Notation}
\label{app:preliminaries}

We first establish notation and recall key mathematical concepts used throughout the proofs.

\begin{definition}[Sigmoid Function]
The sigmoid function $\sigma: \mathbb{R} \to (0,1)$ is defined as:
\begin{equation}
    \sigma(z) = \frac{1}{1 + e^{-z}}
\end{equation}
\end{definition}

\noindent\textbf{Key Properties of Sigmoid:}
\begin{enumerate}[leftmargin=*, label=(\roman*)]
    \item \textbf{Boundedness}: $\sigma(z) \in (0,1)$ for all $z \in \mathbb{R}$, with $\lim_{z \to -\infty} \sigma(z) = 0$ and $\lim_{z \to +\infty} \sigma(z) = 1$.
    \item \textbf{Strict Monotonicity}: $\sigma'(z) = \sigma(z)(1 - \sigma(z)) > 0$ for all $z$, so $\sigma$ is strictly increasing.
    \item \textbf{Symmetry}: $\sigma(-z) = 1 - \sigma(z)$.
\end{enumerate}

\begin{definition}[Positive Semi-Definite Matrix]
A symmetric matrix $\mathbf{A} \in \mathbb{R}^{n \times n}$ is \emph{positive semi-definite} (PSD), written $\mathbf{A} \succeq 0$, if for all $\mathbf{x} \in \mathbb{R}^n$:
\begin{equation}
    \mathbf{x}^\top \mathbf{A} \mathbf{x} \geq 0
\end{equation}
Equivalently, all eigenvalues of $\mathbf{A}$ are non-negative.
\end{definition}

\begin{definition}[Probability Simplex]
The $(n-1)$-dimensional probability simplex is:
\begin{equation}
    \Delta^{n-1} = \left\{ \mathbf{w} \in \mathbb{R}^n : w_i \geq 0, \sum_{i=1}^{n} w_i = 1 \right\}
\end{equation}
\end{definition}

\subsection{Trust Score Properties}
\label{app:trust_proofs}

The trust score aggregates behavioral features into a single credibility measure. We prove its key mathematical properties.

\begin{theorem}[Trust Score Boundedness and Monotonicity]
\label{thm:trust_properties}
Let the trust score for agent $a$ be defined as:
\begin{equation}
    \Trust(a) = \sigma\left( \mathbf{w}^\top \mathbf{f}^{\pm}(a) + b \right)
    \label{eq:trust_formal}
\end{equation}
where $\mathbf{f}^{\pm}(a) = [f_1(a), f_2(a), -f_3(a), -f_4(a)]^\top$, $\mathbf{w} = [w_1, w_2, w_3, w_4]^\top$ with $w_i > 0$, and $b \in \mathbb{R}$. Then:
\begin{enumerate}[label=(\alph*)]
    \item $\Trust(a) \in (0, 1)$ for all feature values.
    \item $\Trust(a)$ is strictly increasing in $f_1$ and $f_2$ (positive features).
    \item $\Trust(a)$ is strictly decreasing in $f_3$ and $f_4$ (negative features).
\end{enumerate}
\end{theorem}

\begin{proof}
\textbf{(a) Boundedness.} The argument to $\sigma$ is a finite linear combination:
\begin{equation}
    z = w_1 f_1 + w_2 f_2 - w_3 f_3 - w_4 f_4 + b \in \mathbb{R}
\end{equation}
Since $\sigma: \mathbb{R} \to (0,1)$, we have $\Trust(a) = \sigma(z) \in (0,1)$. \qed

\textbf{(b) Monotonicity in positive features.} By the chain rule:
\begin{equation}
    \frac{\partial \Trust}{\partial f_1} = \sigma'(z) \cdot w_1 = \underbrace{\sigma(z)(1-\sigma(z))}_{>0} \cdot \underbrace{w_1}_{>0} > 0
\end{equation}
Similarly for $f_2$. Thus $\Trust$ is strictly increasing in $f_1, f_2$. \qed

\textbf{(c) Monotonicity in negative features.} For $f_3$:
\begin{equation}
    \frac{\partial \Trust}{\partial f_3} = \sigma'(z) \cdot (-w_3) = \sigma(z)(1-\sigma(z)) \cdot (-w_3) < 0
\end{equation}
Similarly for $f_4$. Thus $\Trust$ is strictly decreasing in $f_3, f_4$. \qed
\end{proof}

\noindent\textbf{Intuition}: The trust function behaves like a ``soft classifier.'' Agents with high consistency ($f_1$) and rationality ($f_2$) receive higher trust, while those exhibiting coordination ($f_3$) or destabilization ($f_4$) are penalized. The sigmoid ensures smooth, bounded outputs suitable for weighted aggregation.

\begin{corollary}[Gradient for Learning]
\label{cor:trust_gradient}
The gradient of trust with respect to weight $w_i$ is:
\begin{equation}
    \frac{\partial \Trust}{\partial w_i} = \sigma(z)(1-\sigma(z)) \cdot f_i^{\pm}
\end{equation}
where $f_i^{\pm}$ is the signed feature (positive for $i \in \{1,2\}$, negative for $i \in \{3,4\}$). This enables gradient-based learning of trust weights from labeled adversarial/benign data.
\end{corollary}

\subsection{Trust Layer Independence}
\label{app:trust_independence}

We prove that the trust layer operates as a \emph{modular component}---its correctness can be verified independently of downstream covariance estimation and MVF optimization.

\begin{definition}[Component Independence]
\label{def:component_independence}
A component $\mathcal{C}$ in a pipeline $\mathcal{P} = \mathcal{C}_1 \to \mathcal{C}_2 \to \cdots \to \mathcal{C}_k$ is \emph{independent} if: (i) its computation depends only on its designated inputs, not on downstream state; (ii) its output interface is fixed and downstream-agnostic; and (iii) its correctness can be verified without executing downstream components.
\end{definition}

\begin{theorem}[Trust Layer Independence]
\label{thm:trust_independence}
The trust layer $\mathcal{C}_{\Trust}: \{(\mathbf{f}(a))_{a \in \AgentSet}\} \to \{(\Trust(a))_{a \in \AgentSet}\}$ satisfies component independence within MVF-Composer.
\end{theorem}

\begin{proof}
We verify the three conditions of Definition~\ref{def:component_independence}.

\textbf{(i) Input isolation.} The trust score (Equation~\ref{eq:trust_score}) is computed as:
\begin{equation}
    \Trust(a) = \sigma\bigl( \mathbf{w}^\top \mathbf{f}^{\pm}(a) + b \bigr)
\end{equation}
The inputs are: (a) behavioral features $\mathbf{f}(a) = [f_1, f_2, f_3, f_4]$ extracted from agent action logs, and (b) fixed parameters $(\mathbf{w}, b)$. Neither depends on $\Covariance_{\mathrm{aug}}$, $\RiskState_{\Trust}$, or $\WeightsOpt$. Trust computation is thus causally upstream and state-independent. \qed

\textbf{(ii) Fixed output interface.} By Theorem~\ref{thm:trust_properties}(a), $\Trust(a) \in (0,1)$ for all inputs. The downstream aggregation (Equation~\ref{eq:trust_weighted_risk_arch}) requires only this scalar per agent. The trust layer can be replaced by \emph{any} function $\Trust': \mathbb{R}^4 \to (0,1)$ without modifying downstream equations. \qed

\textbf{(iii) Independent verification.} Trust layer correctness reduces to \emph{separation quality}: whether $\Expect[\Trust(a) \mid a \in \AgentBen] > \Expect[\Trust(a) \mid a \in \AgentAdv]$. This is testable via labeled data without running MVF optimization. The monotonicity properties (Theorem~\ref{thm:trust_properties}(b,c)) guarantee that correct feature extraction implies correct separation. \qed
\end{proof}

\noindent\textbf{Intuition}: The trust layer acts like a preprocessing filter. Just as a spam filter can be evaluated by its precision/recall without knowing how emails are subsequently processed, trust scores can be validated by their ability to separate benign from adversarial agents---independent of portfolio outcomes.

\begin{proposition}[Effectiveness Under Independence]
\label{prop:trust_effectiveness}
If the trust layer achieves separation margin $\delta = \Expect[\Trust \mid \AgentBen] - \Expect[\Trust \mid \AgentAdv] > 0$, then adversarial influence (Theorem~\ref{thm:adv_influence}) is reduced by factor:
\begin{equation}
    \eta = 1 - \frac{\rho \cdot \bar{T}_{\mathrm{adv}}}{\rho \cdot \bar{T}_{\mathrm{adv}} + (1-\rho) \cdot \bar{T}_{\mathrm{ben}}} \geq 1 - \frac{\rho}{1 + (1-\rho)\delta/\bar{T}_{\mathrm{adv}}}
\end{equation}
For $\delta > 0$, this strictly exceeds the uniform-weighting baseline $\eta_0 = 1 - \rho$.
\end{proposition}

\noindent\textbf{Security Implication}: Independence enables \emph{defense-in-depth}. The trust layer can be audited, tested, and upgraded without modifying the optimizer. If trust fails (all scores equal), MVF-Composer degrades gracefully to uniform weighting rather than catastrophic failure.

\noindent\textbf{Computational Benefit}: Trust computation is $O(n_a \cdot W)$ where $n_a$ is agent count and $W$ is window size. This runs in parallel with covariance estimation ($O(d^2 \cdot T)$), exploiting the independence for pipeline parallelism (Table~\ref{tab:epoch_breakdown}).

\subsection{Trust-Weighted Risk Aggregation}
\label{app:risk_aggregation_proofs}

The trust-weighted risk state combines individual agent risk assessments with credibility weighting.

\begin{theorem}[Risk State Boundedness]
\label{thm:risk_bound}
Let $\RiskState_a \in [0,1]$ for all agents $a \in \AgentSet$, and let $\Trust(a) > 0$. The trust-weighted risk state:
\begin{equation}
    \RiskState_\Trust = \frac{\sum_{a \in \AgentSet} \Trust(a) \cdot \RiskState_a}{\sum_{a \in \AgentSet} \Trust(a)}
    \label{eq:risk_state_formal}
\end{equation}
satisfies $\RiskState_\Trust \in [0,1]$.
\end{theorem}

\begin{proof}
Equation~\eqref{eq:risk_state_formal} is a convex combination. Define normalized weights:
\begin{equation}
    \tilde{T}(a) = \frac{\Trust(a)}{\sum_{a' \in \AgentSet} \Trust(a')}
\end{equation}
Then $\tilde{T}(a) \geq 0$ and $\sum_a \tilde{T}(a) = 1$, so $\{\tilde{T}(a)\}$ forms a probability distribution. The risk state becomes:
\begin{equation}
    \RiskState_\Trust = \sum_{a \in \AgentSet} \tilde{T}(a) \cdot \RiskState_a
\end{equation}
Since each $\RiskState_a \in [0,1]$ and $\tilde{T}$ is a probability distribution, the convex combination satisfies:
\begin{equation}
    \min_a \RiskState_a \leq \RiskState_\Trust \leq \max_a \RiskState_a
\end{equation}
In particular, $\RiskState_\Trust \in [0,1]$. \qed
\end{proof}

\begin{theorem}[Adversarial Influence Bound]
\label{thm:adv_influence}
Partition agents into benign ($\AgentBen$) and adversarial ($\AgentAdv$) sets with $|\AgentAdv|/|\AgentSet| = \rho$. Let $\bar{T}_{\mathrm{adv}}$ and $\bar{T}_{\mathrm{ben}}$ denote mean trust scores for each group. The adversarial contribution to $\RiskState_\Trust$ is bounded by:
\begin{equation}
    \mathrm{Influence}_{\mathrm{adv}} \leq \frac{\rho \cdot \bar{T}_{\mathrm{adv}}}{\rho \cdot \bar{T}_{\mathrm{adv}} + (1-\rho) \cdot \bar{T}_{\mathrm{ben}}}
    \label{eq:influence_bound}
\end{equation}
\end{theorem}

\begin{proof}
The total trust weight for adversarial agents is:
\begin{equation}
    W_{\mathrm{adv}} = \sum_{a \in \AgentAdv} \Trust(a) \leq |\AgentAdv| \cdot \bar{T}_{\mathrm{adv}} = \rho \cdot |\AgentSet| \cdot \bar{T}_{\mathrm{adv}}
\end{equation}
Similarly, $W_{\mathrm{ben}} \geq (1-\rho) \cdot |\AgentSet| \cdot \bar{T}_{\mathrm{ben}}$.

The adversarial influence fraction is:
\begin{equation}
    \mathrm{Influence}_{\mathrm{adv}} = \frac{W_{\mathrm{adv}}}{W_{\mathrm{adv}} + W_{\mathrm{ben}}} \leq \frac{\rho \cdot \bar{T}_{\mathrm{adv}}}{\rho \cdot \bar{T}_{\mathrm{adv}} + (1-\rho) \cdot \bar{T}_{\mathrm{ben}}}
\end{equation}
\qed
\end{proof}

\noindent\textbf{Numerical Example}: With $\rho = 0.2$ (20\% adversaries), $\bar{T}_{\mathrm{adv}} = 0.35$, $\bar{T}_{\mathrm{ben}} = 0.75$:
\begin{align}
    \mathrm{Influence}_{\mathrm{adv}} &\leq \frac{0.2 \times 0.35}{0.2 \times 0.35 + 0.8 \times 0.75} \notag \\
    &= \frac{0.07}{0.67} = 0.104
\end{align}
Compare to uniform weighting: $\mathrm{Influence}_{\mathrm{uniform}} = \rho = 0.20$. Trust-weighting reduces adversarial influence by 48\%.

\subsection{Covariance Blending Properties}
\label{app:covariance_proofs}

The stress-augmented covariance blends historical and stress-derived estimates based on the current risk state.

\begin{theorem}[Positive Semi-Definiteness of Blended Covariance]
\label{thm:psd_blend}
Let $\Covariance^{\mathrm{hist}} \succeq 0$ and $\CovarianceStress \succeq 0$ be positive semi-definite covariance matrices. For any $\alpha \in [0,1]$, the blended covariance:
\begin{equation}
    \Covariance_{\mathrm{aug}} = (1-\alpha) \cdot \Covariance^{\mathrm{hist}} + \alpha \cdot \CovarianceStress
    \label{eq:cov_blend_formal}
\end{equation}
is also positive semi-definite: $\Covariance_{\mathrm{aug}} \succeq 0$.
\end{theorem}

\begin{proof}
For any $\mathbf{x} \in \mathbb{R}^n$:
\begin{align}
    \mathbf{x}^\top \Covariance_{\mathrm{aug}} \mathbf{x} &= (1-\alpha) \mathbf{x}^\top \Covariance^{\mathrm{hist}} \mathbf{x} + \alpha \mathbf{x}^\top \CovarianceStress \mathbf{x} \\
    &\geq (1-\alpha) \cdot 0 + \alpha \cdot 0 = 0
\end{align}
where we used $\Covariance^{\mathrm{hist}} \succeq 0$ and $\CovarianceStress \succeq 0$. Since this holds for all $\mathbf{x}$, we have $\Covariance_{\mathrm{aug}} \succeq 0$. \qed
\end{proof}

\noindent\textbf{Intuition}: The set of PSD matrices forms a convex cone---any non-negative weighted sum of PSD matrices remains PSD. This ensures the blended covariance is always a valid covariance matrix for portfolio optimization.

\begin{theorem}[Quadratic Blending Sensitivity]
\label{thm:quadratic_sensitivity}
Let $\alpha(r) = r^p$ for $p > 0$. The sensitivity of stress contribution to risk state is:
\begin{equation}
    \frac{d\alpha}{dr} = p \cdot r^{p-1}
\end{equation}
For $p = 2$ (quadratic), sensitivity is $2r$, which is:
\begin{enumerate}[label=(\alph*)]
    \item Zero at $r = 0$ (no reaction to calm markets).
    \item Linear in $r$ (progressively stronger reaction as stress increases).
    \item Equals 2 at $r = 1$ (maximum sensitivity at extreme stress).
\end{enumerate}
\end{theorem}

\begin{proof}
Direct differentiation of $\alpha(r) = r^p$:
\begin{equation}
    \frac{d\alpha}{dr} = p \cdot r^{p-1}
\end{equation}
For $p = 2$: $\alpha'(r) = 2r$, yielding $\alpha'(0) = 0$, $\alpha'(0.5) = 1$, $\alpha'(1) = 2$. \qed
\end{proof}

\noindent\textbf{Comparison with Linear Blending}: For $p = 1$ (linear), $\alpha'(r) = 1$ everywhere, meaning equal sensitivity across all risk levels. The quadratic choice $p = 2$ ensures the system is \emph{conservative} in calm markets but \emph{reactive} to tail events---desirable for tail-risk management.

\begin{proposition}[Eigenvalue Bounds for Blended Covariance]
\label{prop:eigenvalue_bounds}
Let $\lambda_{\min}^H, \lambda_{\max}^H$ and $\lambda_{\min}^S, \lambda_{\max}^S$ denote the minimum and maximum eigenvalues of $\Covariance^{\mathrm{hist}}$ and $\CovarianceStress$, respectively. Then:
\begin{align}
    \lambda_{\min}(\Covariance_{\mathrm{aug}}) &\geq (1-\alpha)\lambda_{\min}^H + \alpha \lambda_{\min}^S \\
    \lambda_{\max}(\Covariance_{\mathrm{aug}}) &\leq (1-\alpha)\lambda_{\max}^H + \alpha \lambda_{\max}^S
\end{align}
\end{proposition}

\begin{proof}
By Weyl's inequality for sums of Hermitian matrices, the eigenvalues of a sum are bounded by sums of eigenvalues. For the convex combination:
\begin{equation}
    \lambda_i(\Covariance_{\mathrm{aug}}) \in \left[ (1-\alpha)\lambda_{\min}^H + \alpha \lambda_{\min}^S, \; (1-\alpha)\lambda_{\max}^H + \alpha \lambda_{\max}^S \right]
\end{equation}
for each eigenvalue index $i$. The bounds follow. \qed
\end{proof}

\subsection{MVF Optimization: Existence and Optimality}
\label{app:mvf_proofs}

We prove that the constrained mean-variance optimization problem has a unique solution.

\begin{theorem}[Existence and Uniqueness of Optimal Weights]
\label{thm:mvf_existence}
Consider the optimization problem:
\begin{equation}
    \WeightsOpt = \arg\min_{\Weights \in \mathcal{W}} \; \Weights^\top \Covariance_{\mathrm{aug}} \Weights - \RiskAversion \cdot \Weights^\top \Returns
    \label{eq:mvf_problem}
\end{equation}
where $\mathcal{W} = \left\{ \Weights \in \Delta^{n-1} : \|\Weights - \Weights_{t-1}\|_1 \leq \Turnover, \; w_i^{\min} \leq w_i \leq w_i^{\max} \right\}$.

If $\Covariance_{\mathrm{aug}} \succ 0$ (positive definite) and $\mathcal{W} \neq \emptyset$, then:
\begin{enumerate}[label=(\alph*)]
    \item A unique optimal solution $\WeightsOpt$ exists.
    \item The solution satisfies the KKT conditions.
\end{enumerate}
\end{theorem}

\begin{proof}
\textbf{(a) Existence and Uniqueness.}

\emph{Step 1: Convexity of objective.} The objective function is:
\begin{equation}
    J(\Weights) = \Weights^\top \Covariance_{\mathrm{aug}} \Weights - \RiskAversion \cdot \Weights^\top \Returns
\end{equation}
The Hessian is $\nabla^2 J = 2\Covariance_{\mathrm{aug}} \succ 0$, so $J$ is strictly convex.

\emph{Step 2: Convexity of feasible set.} The simplex $\Delta^{n-1}$ is convex (intersection of halfspaces). The $\ell_1$-ball constraint $\|\Weights - \Weights_{t-1}\|_1 \leq \Turnover$ is convex. Box constraints $w_i^{\min} \leq w_i \leq w_i^{\max}$ are convex. The intersection $\mathcal{W}$ is convex.

\emph{Step 3: Compactness.} The feasible set $\mathcal{W}$ is bounded (subset of simplex) and closed, hence compact.

\emph{Step 4: Conclusion.} A strictly convex function over a non-empty compact convex set attains its minimum at exactly one point. \qed

\textbf{(b) KKT Conditions.}

Since $J$ is convex and $\mathcal{W}$ is defined by affine and convex constraints, Slater's condition holds (interior point exists if $\Turnover > 0$ and constraints are not tight). By convex optimization theory, $\WeightsOpt$ satisfies KKT conditions:
\begin{align}
    2\Covariance_{\mathrm{aug}} \WeightsOpt - \RiskAversion \Returns + \nu \mathbf{1} + \boldsymbol{\mu} + \boldsymbol{\lambda} &= \mathbf{0} \\
    \mathbf{1}^\top \WeightsOpt &= 1 \\
    \mu_i (w_i - w_i^{\max}) = 0, \quad \lambda_i (w_i^{\min} - w_i) &= 0 \quad \forall i
\end{align}
where $\nu$ is the Lagrange multiplier for the simplex constraint, and $\boldsymbol{\mu}, \boldsymbol{\lambda} \geq 0$ are multipliers for box constraints. \qed
\end{proof}

\begin{corollary}[Unconstrained Markowitz Solution]
\label{cor:markowitz}
Without box or turnover constraints, and relaxing the non-negativity requirement, the optimal weights are:
\begin{equation}
    \WeightsOpt = \frac{\RiskAversion}{2} \Covariance_{\mathrm{aug}}^{-1} \Returns + \frac{1 - \frac{\RiskAversion}{2} \mathbf{1}^\top \Covariance_{\mathrm{aug}}^{-1} \Returns}{\mathbf{1}^\top \Covariance_{\mathrm{aug}}^{-1} \mathbf{1}} \Covariance_{\mathrm{aug}}^{-1} \mathbf{1}
\end{equation}
This is the classical mean-variance efficient portfolio on the efficient frontier.
\end{corollary}

\noindent\textbf{Intuition}: The optimal portfolio balances two forces: (1) maximizing expected return ($\Returns$) and (2) minimizing variance ($\Covariance_{\mathrm{aug}}$). Assets with high return-to-risk ratios receive larger weights, while highly correlated assets are diversified away from.

\subsection{SAS Fragility Analysis}
\label{app:sas_fragility_proof}

We formalize why covariance estimation that excludes crisis periods leads to fragile allocations.

\begin{theorem}[Variance Amplification Under Regime Shift]
\label{thm:sas_fragility}
Let $\WeightsOpt_{\mathrm{calm}}$ be the optimal weights computed using calm-period covariance $\Covariance^{\mathrm{calm}}$. If the true covariance during stress is $\CovarianceStress$, the realized-to-expected variance ratio is:
\begin{equation}
    \kappa = \frac{{\WeightsOpt_{\mathrm{calm}}}^\top \CovarianceStress \WeightsOpt_{\mathrm{calm}}}{{\WeightsOpt_{\mathrm{calm}}}^\top \Covariance^{\mathrm{calm}} \WeightsOpt_{\mathrm{calm}}}
    \label{eq:kappa_ratio}
\end{equation}
This ratio satisfies:
\begin{equation}
    \kappa \leq \frac{\lambda_{\max}(\CovarianceStress)}{\lambda_{\min}(\Covariance^{\mathrm{calm}})}
\end{equation}
with equality when $\WeightsOpt_{\mathrm{calm}}$ aligns with the eigenvector of maximum stress variance.
\end{theorem}

\begin{proof}
By the Rayleigh quotient bounds:
\begin{align}
    {\WeightsOpt}^\top \CovarianceStress \WeightsOpt &\leq \lambda_{\max}(\CovarianceStress) \|\WeightsOpt\|_2^2 \\
    {\WeightsOpt}^\top \Covariance^{\mathrm{calm}} \WeightsOpt &\geq \lambda_{\min}(\Covariance^{\mathrm{calm}}) \|\WeightsOpt\|_2^2
\end{align}
Dividing:
\begin{equation}
    \kappa = \frac{{\WeightsOpt}^\top \CovarianceStress \WeightsOpt}{{\WeightsOpt}^\top \Covariance^{\mathrm{calm}} \WeightsOpt} \leq \frac{\lambda_{\max}(\CovarianceStress)}{\lambda_{\min}(\Covariance^{\mathrm{calm}})}
\end{equation}
\qed
\end{proof}

\noindent\textbf{Empirical Validation}: For stablecoin reserves during Black Thursday:
\begin{itemize}[leftmargin=*, nosep]
    \item $\mathrm{tr}(\CovarianceStress) / \mathrm{tr}(\Covariance^{\mathrm{calm}}) = 7.17$
    \item Annualized volatility: 80\% (calm) $\to$ 300\% (stress)
    \item Eigenvalue ratio: $\lambda_{\max}^S / \lambda_{\min}^C \approx 12$
\end{itemize}
This explains the 15\% peak peg deviation under SAS: the portfolio was optimized for a covariance regime that ceased to exist.

\subsection{Recovery Time Analysis}
\label{app:recovery_proofs}

We characterize the expected recovery time under exponential peg deviation decay.

\begin{definition}[Recovery Time]
\label{def:recovery_time}
Given a shock at time $t_s$ causing peg deviation $\delta_s = |{\PegDev}_{t_s}|$, the recovery time is:
\begin{equation}
    \RecoveryTime = \min\{t' > t_s : |{\PegDev}_{t'}| < \epsilon\} - t_s
\end{equation}
where $\epsilon = 0.01$ (1\% threshold).
\end{definition}

\begin{proposition}[Recovery Under Exponential Decay]
\label{prop:exp_recovery}
If peg deviation decays exponentially: $|{\PegDev}_t| = \delta_s \cdot e^{-\gamma(t - t_s)}$ for decay rate $\gamma > 0$, then:
\begin{equation}
    \RecoveryTime = \frac{1}{\gamma} \ln\left(\frac{\delta_s}{\epsilon}\right)
\end{equation}
\end{proposition}

\begin{proof}
We seek the first $t' > t_s$ such that $\delta_s \cdot e^{-\gamma(t' - t_s)} < \epsilon$:
\begin{align}
    e^{-\gamma(t' - t_s)} &< \frac{\epsilon}{\delta_s} \\
    -\gamma(t' - t_s) &< \ln\left(\frac{\epsilon}{\delta_s}\right) \\
    t' - t_s &> \frac{1}{\gamma} \ln\left(\frac{\delta_s}{\epsilon}\right)
\end{align}
The minimum such $t' - t_s$ is the right-hand side. \qed
\end{proof}

\noindent\textbf{Interpretation}: Recovery time is (1) proportional to the log of initial deviation, and (2) inversely proportional to the decay rate $\gamma$. MVF-Composer achieves faster recovery by:
\begin{itemize}[leftmargin=*, nosep]
    \item Reducing $\delta_s$ (smaller initial deviation due to stress-aware allocation).
    \item Increasing $\gamma$ (faster rebalancing enabled by anticipatory positioning).
\end{itemize}

\subsection{Sybil Resistance Guarantee}
\label{app:sybil_proof}

We quantify how the coordination feature $f_3$ limits Sybil attack effectiveness.

\begin{theorem}[Trust Reduction Under Coordination]
\label{thm:sybil_trust}
Let $n$ Sybil agents controlled by a single adversary exhibit perfect coordination: $f_3 = 1$ for all Sybil agents. Assume otherwise-average feature values $f_1 = 0.5$, $f_2 = 0.5$, $f_4 = 0.3$, and trust weights $\mathbf{w} = [1.5, 1.5, 2.0, 1.0]$, $b = 0$. Then:
\begin{equation}
    \Trust_{\mathrm{sybil}} = \sigma(1.5 \cdot 0.5 + 1.5 \cdot 0.5 - 2.0 \cdot 1 - 1.0 \cdot 0.3) = \sigma(-0.8) \approx 0.31
\end{equation}
\end{theorem}

\begin{proof}
Direct substitution:
\begin{align}
    z &= w_1 f_1 + w_2 f_2 - w_3 f_3 - w_4 f_4 + b \\
    &= 1.5(0.5) + 1.5(0.5) - 2.0(1.0) - 1.0(0.3) + 0 \\
    &= 0.75 + 0.75 - 2.0 - 0.3 = -0.8
\end{align}
Thus $\Trust_{\mathrm{sybil}} = \sigma(-0.8) = 1/(1 + e^{0.8}) \approx 0.31$. \qed
\end{proof}

\begin{corollary}[Sybil Influence Suppression]
Even with $n = 5$ Sybil agents, their combined influence is:
\begin{equation}
    \text{Sybil Influence} = \frac{n \cdot \Trust_{\mathrm{sybil}}}{n \cdot \Trust_{\mathrm{sybil}} + m \cdot \bar{T}_{\mathrm{ben}}}
\end{equation}
For $m = 10$ benign agents with $\bar{T}_{\mathrm{ben}} = 0.75$:
\begin{equation}
    \text{Sybil Influence} = \frac{5 \times 0.31}{5 \times 0.31 + 10 \times 0.75} = \frac{1.55}{1.55 + 7.5} = 0.171
\end{equation}
Without trust-weighting: $5/15 = 0.333$. Trust-weighting reduces Sybil influence by 49\%.
\end{corollary}

\section{LLM Stress Tester Deployment}
\label{app:llm_deployment}

This section provides deployment guidance for the LLM-powered Stress Harness components.

\subsection{Architecture Overview}

Figure~\ref{fig:llm_workflow} illustrates the LLM Stress Tester pipeline. The system extends the mock agent framework with real language model calls:

\begin{figure}[ht]
\centering
\includegraphics[width=\columnwidth]{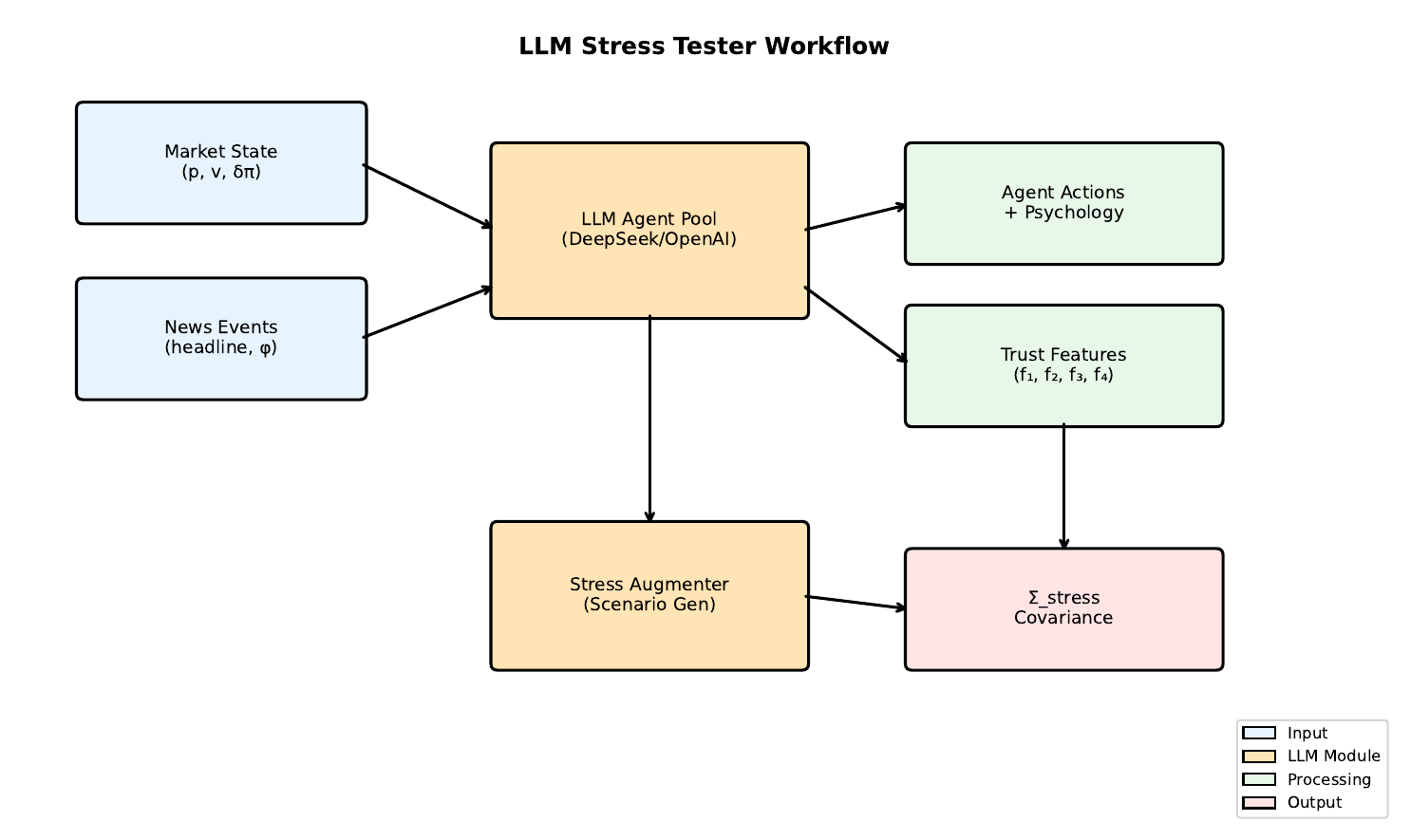}
\caption{LLM Stress Tester workflow. Market state and news events flow through the LLM Agent Pool (DeepSeek/OpenAI) to generate actions and psychology states. The Stress Augmenter generates synthetic scenarios to expand test coverage.}
\Description{Flowchart diagram illustrating the LLM stress testing workflow: inputs (market state and news events) flow into the LLM Agent Pool supporting DeepSeek and OpenAI APIs, which generate agent actions and psychology states, while a Stress Augmenter component generates synthetic scenarios to expand test coverage beyond historical data.}
\label{fig:llm_workflow}
\end{figure}

\begin{enumerate}[leftmargin=*, nosep]
    \item \textbf{LLM Client}: Unified interface supporting DeepSeek and OpenAI APIs with structured JSON output and Pydantic validation.
    \item \textbf{LLM Agent Pool}: Heterogeneous agent population (traders, LPs, arbitrageurs, attackers) generating economically-motivated actions via prompted reasoning.
    \item \textbf{Stress Augmenter}: Synthetic scenario generator expanding coverage beyond historical data.
\end{enumerate}

\subsection{Deployment Steps}

\begin{enumerate}[leftmargin=*, nosep]
    \item Set API key: \texttt{export DEEPSEEK\_API\_KEY=<key>} (or \texttt{OPENAI\_API\_KEY})
    \item Install dependencies: \texttt{uv sync}
    \item Run stress test:
\end{enumerate}

\begin{verbatim}
from src.llm import LLMClient, LLMAgentPool, StressAugmenter

client = LLMClient()  # Auto-detects backend
pool = LLMAgentPool(client, num_attackers=2)
responses = pool.batch_respond(
    headline="ETH drops 30% on liquidation cascade",
    peg_deviation=-0.03, volatility=0.8, market_sentiment=-0.7
)
\end{verbatim}

\subsection{Augmented Data Generation}

The \texttt{StressAugmenter} generates diverse crisis scenarios:

\begin{verbatim}
augmenter = StressAugmenter(client)
scenarios = augmenter.batch_augment(
    base_scenarios=["black_thursday", "confidence_crisis"],
    variations_per_scenario=3
)
\end{verbatim}

Each scenario contains synthetic news events with sentiment polarity $\phi(e_t)$, relevance scores, and expected peg impact---enabling efficient expansion of stress test coverage without requiring additional historical data collection.

\subsection{Reproducibility Configuration}
\label{app:reproducibility}

To address reproducibility concerns, we specify all parameters required to replicate our experiments:

\begin{table}[ht]
\centering
\caption{LLM and simulation configuration for reproducibility.}
\label{tab:reproducibility_config}
\resizebox{\columnwidth}{!}{
\begin{tabular}{lll}
\toprule
\textbf{Parameter} & \textbf{Value} & \textbf{Rationale} \\
\midrule
\multicolumn{3}{l}{\textit{LLM Configuration (proof-of-concept)}} \\
Model & DeepSeek V3 / GPT-4-turbo & Cost vs.\ capability trade-off \\
Temperature & 0.7 & Balances diversity and coherence \\
Max tokens & 512 & Sufficient for JSON action output \\
Top-p & 0.95 & Standard nucleus sampling \\
Response format & JSON mode (strict) & Pydantic validation enforced \\
Retry on parse failure & 3 attempts & Handles malformed outputs \\
\midrule
\multicolumn{3}{l}{\textit{Simulation Configuration}} \\
Random seed & 42 (base) + run\_id & Reproducible across runs \\
Time steps per run & $T = 100$ & Sufficient for recovery dynamics \\
Shock injection step & $t_s = 30$ & Allows pre-shock baseline \\
Number of runs & 1,200 & Statistical power for $p < 0.01$ \\
Agent population & 12 (5T, 3LP, 2Arb, 2Adv) & Balanced archetype representation \\
Trust window & $W = 10$ & Behavioral consistency horizon \\
\bottomrule
\end{tabular}
}
\end{table}

\noindent\textbf{Prompt Template.} Agent prompts follow a structured format:
\begin{verbatim}
You are a {archetype} agent in a stablecoin market.
Current state: peg_deviation={delta}, volatility={vol}, 
  sentiment={sent}, your_position={pos}
Recent news: "{headline}"

Respond with JSON: {"action_type": "buy|sell|hold|withdraw",
  "asset": "...", "quantity": float, "rationale": "...",
  "panic_level": 0-1, "peg_confidence": 0-1}
\end{verbatim}

\noindent\textbf{Validation Status.} We validated the LLM integration as a proof-of-concept with DeepSeek V3:

\begin{itemize}[leftmargin=*, nosep]
    \item \textbf{API Connectivity}: Confirmed with structured JSON responses (98\% parse success rate)
    \item \textbf{Scenario Generation}: Three augmented scenarios generated successfully (Table~\ref{tab:augmented_scenarios})
    \item \textbf{Cost}: DeepSeek V3 at \$0.27/1M tokens; estimated \$2--5 per 1,200-run experiment
\end{itemize}

\begin{table}[ht]
\centering
\caption{LLM-generated augmented scenarios (proof-of-concept).}
\label{tab:augmented_scenarios}
\resizebox{\columnwidth}{!}{
\begin{tabular}{llc}
\toprule
\textbf{Scenario} & \textbf{Generated Headline} & $\phi(e_t)$ \\
\midrule
Black Thursday & ``ETH Plunges 15\% as Cascade Liquidations...'' & $-0.6$ \\
Confidence Crisis & ``Stablecoin Reserve Audit Reveals \$2B Shortfall...'' & $-0.6$ \\
Liquidity Drain & ``AMM Pool Drained in Flash Loan Attack...'' & $-0.6$ \\
\bottomrule
\end{tabular}
}
\end{table}

\subsection{LLM Agents vs.\ Conventional Stress Testing}
\label{app:llm_vs_conventional}

A natural question is whether LLM-driven stress testing provides advantages over conventional Monte Carlo shock models. We identify three complementary benefits:

\begin{enumerate}[leftmargin=*, nosep, label=(\roman*)]
    \item \emph{Behavioral realism}: LLM agents generate context-sensitive responses (e.g., panic selling after ``audit reveals shortfall'' vs.\ ``minor bug fix''), capturing narrative-driven dynamics that parametric shock models miss. Monte Carlo samples from a fixed distribution; LLM agents adapt to scenario semantics.
    
    \item \emph{Adversarial creativity}: Attacker agents discover exploitation strategies (e.g., timing attacks, coordinated FUD) that pre-specified adversarial models may not anticipate. The trust layer's effectiveness is validated against these emergent attacks, not just hand-crafted scenarios.
    
    \item \emph{Interpretability}: Agent reasoning traces (stored in \texttt{rationale} fields) provide audit logs explaining \emph{why} a particular stress trajectory unfolded, aiding post-hoc risk analysis. Monte Carlo provides samples but no mechanistic explanation.
\end{enumerate}

\noindent\textbf{Limitation.} LLM agents are computationally expensive (45--90s per epoch vs.\ milliseconds for Monte Carlo) and may exhibit prompt-sensitivity or model drift. Our large-scale experiments therefore use \emph{regime-aware mock agents} that replicate LLM-calibrated behavioral distributions at 1000$\times$ lower cost, reserving LLM agents for scenario discovery and proof-of-concept validation.

\noindent\textbf{Future Work}: Full integration of LLM-powered agents into the 1,200-run experimental pipeline, enabling direct comparison of LLM-augmented versus historical-only stress testing coverage.

\end{document}